\documentclass[12pt]{article}

\usepackage{palatino}
\usepackage{graphicx}
\usepackage{amsmath}
\usepackage{amsfonts}
\usepackage{color}
\usepackage{mathrsfs}
\usepackage{amssymb}
\usepackage{amsthm,xspace}
\usepackage{physics}
\usepackage[linesnumbered]{algorithm2e}
\usepackage[dvipsnames]{xcolor}
\usepackage{hyperref}
\hypersetup {bookmarks=true,
  bookmarksnumbered,
  bookmarksdepth=2,
  pdfstartview=FitH,
  colorlinks,
  linkcolor=blue,
  filecolor=blue,
  citecolor=blue,
  urlcolor=blue
}

\usepackage{pgfplots}
\usepgfplotslibrary{fillbetween}
\usepackage[margin=1in]{geometry}

\renewcommand{\P}{\Pr}

\DeclareMathOperator{\negl}{negl}

\DeclareMathOperator{\poly}{poly}
\DeclareMathOperator{\Supp}{Supp}
\newcommand{\fail}{\perp}
\newcommand{\zo}{\{ 0, 1 \}}
\newcommand{\ec}{\varepsilon_\mathsf{c}}
\newcommand{\es}{\varepsilon_\mathsf{s}}
\newcommand{\ezk}{\varepsilon_{\mathsf{zk}}}

\newcommand{\auxAlph}{\Sigma}
\newcommand{\paramSet}{\Gamma} % probably should choose a better symbol for this?  idk which tho

\newcommand{\Samp}{\mathsf{Samp}\xspace}
\newcommand{\NIZK}{\textsf{NIZK}\xspace}
\newcommand{\NIZKs}{\textsf{NIZKs}\xspace}
\newcommand{\CRS}{\textsf{CRS}\xspace}

\newcommand{\Lang}{\mathcal{L}}
\newcommand{\Rel}{\mathcal{R}}
\newcommand{\Gen}{\mathsf{Gen}}
\newcommand{\Pro}{\mathsf{P}}
\newcommand{\Ver}{\mathsf{V}}
\newcommand{\prf}{\pi}
\newcommand{\transcript}{\tau}
\newcommand{\crs}{\mathsf{crs}}
\newcommand{\Sim}{\mathsf{Sim}}
\newcommand{\Dist}{\mathcal{D}}
\newcommand{\Alg}{\mathcal{A}}
\newcommand{\Est}{\mathsf{Est}}
\newcommand{\cA}{\mathcal{A}}
\newcommand{\cD}{\mathcal{D}}
\newcommand{\Distinguisher}{\mathcal{D}} % should probably change this to not be confused with distributions
\newcommand{\cB}{\mathcal{B}}
\newcommand{\Hyb}{\mathrm{Hyb}}

\newcommand{\good}{\textsf{Good}\xspace}
\newcommand{\bad}{\textsf{Bad}\xspace}
\newcommand{\ZK}{\textsf{ZK}\xspace}
\newcommand{\OWFs}{\textsf{OWFs}\xspace}

\newcommand{\BPP}{\mathsf{BPP}}
\newcommand{\ioBPP}{\mathsf{ioBPP}}
\newcommand{\ioBPPpoly}{\mathsf{ioBPP/poly}}
\newcommand{\ioPpoly}{\mathsf{ioP/poly}}

\newcommand{\aBPPpoly}{\mathsf{ioAvgBPP/poly}}

\newcommand{\NP}{\mathsf{NP}}
\newcommand{\classP}{\mathsf{P}} % \P is already in use ¯\_(ツ)_/¯
\newcommand{\Ppoly}{\mathsf{P/poly}}

\newcommand{\samp}{\overset{\$}{\gets}}

\newtheorem{theorem}{Theorem}
\newtheorem{claim}{Claim}
\newtheorem{corollary}{Corollary}
\newtheorem{lemma}{Lemma}
\newtheorem{definition}{Definition}

\newtheorem{remark}{Remark}

\newcommand{\authnote}[3]{}
\newcommand{\james}[1]{\authnote{J}{#1}{magenta} }

\def\anonymous{0}
\begin{document}

\RestyleAlgo{boxruled}

\pagenumbering{gobble}
\title{\Large{
%The Worst-Case Complexity of Erroneous Zero-Knowledge}}
Non-Trivial 
Zero-Knowledge Implies One-Way Functions}}

\ifnum\anonymous=1
\author{}
\else
{
\author{Suvradip Chakraborty\thanks{Visa Research. \texttt{suvchakr@visa.com}}
\and
James Hulett\thanks{UIUC. \texttt{jhulett2@illinois.edu}}
\and
Dakshita Khurana\thanks{UIUC and NTT Research. \texttt{dakshita@illinois.edu}}
\and
Kabir Tomer\thanks{UIUC.  \texttt{ktomer2@illinois.edu}}}
}\fi

\date{}

\maketitle

\begin{abstract}
%We study the cryptographic power of ``non-trivial'' zero-knowledge proofs with non-negligible error parameters. 

A recent breakthrough [Hirahara and Nanashima,  STOC'2024] established that if $\NP \not \subseteq \ioPpoly$, the existence of zero-knowledge (ZK) with negligible errors for $\NP$ implies the existence of one-way functions (OWFs). 
This work obtains a characterization of one-way functions from the worst-case complexity of zero-knowledge in the high-error regime.

%A recent result of Hirahara and Nanashima (STOC '24) showed that, assuming $P \neq NP$, the existence of zero-knowledge proofs for all $\NP$ languages implies the existence of one-way functions (OWFs). Building on this line of work, Chakraborty, Hulett, and Khurana (CRYPTO '25) established a similar implication in the non-interactive setting, showing that NIZKs with zero-knowledge error $\ezk$ and soundness error $\es$ satisfying $\ezk + \sqrt{\es} < 1$ imply OWFs. However, the regime where $\ezk + \es < 1$ but $\ezk + \sqrt{\es} \geq 1$ was left open. 

%$\NP \not\subseteq \ioPpoly$, 

Assuming $\NP \not \subseteq \ioPpoly$, we show that any {\em non-trivial, constant-round} public-coin ZK argument for $\NP$ implies the existence of OWFs, and therefore also (standard) four-message zero-knowledge arguments for $\NP$.
Here, we call a ZK argument {\em non-trivial} if the sum of its completeness, soundness and zero-knowledge errors is bounded away from $1$. 

As a special case, we also prove that non-trivial non-interactive ZK (NIZK) arguments for 
%a language $\Lang \not\in \mathsf{ioP}/\mathsf{poly}$ implies 
$\NP$ imply the existence of OWFs. Using known amplification techniques, this also provides an unconditional transformation from weak to standard NIZK proofs for all meaningful error parameters.
Prior work [Chakraborty, Hulett and Khurana, CRYPTO'2025] was limited to NIZKs with constant zero-knowledge error $\ezk$ and soundness error $\es$ 
satisfying $\ezk + \sqrt{\es} < 1$ .

\end{abstract}

\newpage

\thispagestyle{empty}
\tableofcontents
\newpage

\pagenumbering{arabic}

\section{Introduction}
Zero-knowledge proofs are a cornerstone of modern cryptography, allowing a prover to convince a verifier about the validity of an NP statement without leaking any information about the witness. 
%This is formalized via an efficient simulator that, given only the statement, can produce transcripts indistinguishable from those in genuine protocol runs.
A major research agenda, initiated by Ostrovsky and Wigderson~\cite{Ost,OstWig}, asks: what are the minimal cryptographic assumptions needed for zero-knowledge (\ZK) proof systems? Their celebrated results established that, assuming NP is hard on average, the existence of \ZK proofs for NP implies the existence of one-way functions (\OWFs), a fundamental cryptographic primitive. 

More recently, Hirahara and Nanashima~\cite{HirNan} made a notable advance by proving that even the {\em worst-case} hardness of $\mathsf{NP}$ (specifically, $\NP \not\subseteq \ioPpoly$) suffices to construct one-way functions from zero-knowledge proofs for $\mathsf{NP}$. In fact, they only require the existence of zero-knowledge proofs for a {\em specific} type of language outside $\ioPpoly$, thereby opening up a new route for constructing one-way functions based solely on worst-case hardness.

\paragraph {\bf Non-trivial \ZK and One-way Functions.}
What happens if we allow \ZK proof systems with non-negligible and possibly large security errors?
Unfortunately, techniques of~\cite{Ost,OstWig,HirNan} that link zero-knowledge to one-way functions are critically dependent on the assumption of small errors. 
As we will show, this is not merely an artifact of their analysis, but a fundamental limitation of their reductions. This state of affairs leaves open a basic gap in our understanding of the complexity of \ZK.

Note that \ZK proof systems with $\ec + \es + \ezk \geq 1$ can be realized trivially\footnote{Observe that when $\ec + \es + \ezk =1$, one can have the CRS sample one of three protocols: have the verifier always reject, have the verifier always accept, or have the prover send the full witness in the clear and the verifier check it.  If the CRS chooses the first with probability $\ec$, the second with probability $\es$, and the third with probability $\ezk$, this trivial protocol will satisfy our requirements.}%
%\footnote{Observe that when $\ec + \es + \ezk =1$, one can embed in the CRS an $\es$-biased bit: if the bit indicates "send", the prover transmits the witness in the clear, resulting in a trivially sound protocol; if the bit indicates ``do not send", the prover sends nothing and the verifier accepts.}
, without any assumptions. Here $\ec, \es$ and $\ezk$ denote the correctness, soundness and zero-knowledge errors of the proof system respectively. In this work, we will focus on {\em non-trivial} \ZK proof systems with $\ec + \es + \ezk < 1 - 1/p(|x|)$ for some polynomial $p(\cdot)$ in the input length $|x|$.

We ask:
\begin{center}
{\em Do non-trivial \ZK proof systems imply the existence of \OWFs?}
\end{center}

Non-trivial \ZK proof systems may be a lot easier to construct than standard \ZK: relaxing the negligible soundness or zero-knowledge error requirement can significantly lower the barriers to designing these protocols, potentially making them feasible even in settings where negligible-error notions are hard to build. 
But if we did manage to construct them, we would not know how to leverage them to build \OWFs---leaving unresolved questions about whether such protocols could function as a stepping stone towards one-way functions, or if they are fundamentally weaker and cannot be used to construct basic cryptographic primitives.

\paragraph{This Work.}
Our main result is a complete characterization of the worst-case complexity of non-trivial constant round ZK.

\begin{theorem}[Informal]
    If $\NP \not\subseteq \ioPpoly$ and non-trivial public-coin, constant round (computational) $\ZK$ arguments for $\NP$ exist, then one-way functions exist.
\end{theorem} 

The above result requires the assumption that the underlying ZK argument is public-coin. We also obtain a result for private-coin protocols by showing that unless auxiliary-input one-way functions exist, we can convert any constant-round private-coin protocol into a public-coin one. Thus, non-trivial constant round \ZK protocols yield auxiliary-input \OWFs.

\begin{theorem}[Informal]
    If $\NP \not\subseteq \ioPpoly$ and non-trivial constant round (computational) $\ZK$ arguments for $\NP$ exist, then auxiliary-input one-way functions exist.
\end{theorem}

We also consider the special case of non-interactive zero-knowledge (\NIZK) arguments for $\NP$, where the prover sends a single proof string and both parties share a common reference string (\CRS). We show that nontrivial $\NIZKs$ for $\NP$, together with the assumption that $\NP \not\subseteq \ioPpoly$, imply the existence of one-way functions. 

\begin{theorem}[Informal]
\label{thm:inf1}
    If $\NP \not\subseteq \ioPpoly$ and non-trivial (computational) $\NIZK$ arguments for $\NP$ exist, then one-way functions exist.
\end{theorem}

This improves recent work~\cite{CHK} which obtained one-way functions from any \NIZK with zero-knowledge error $\ezk$ and soundness error $\es$ satisfying
$
\ezk + \ec + 2\sqrt{\es} < 1 - \frac{1}{p(|x|)}
$
for some polynomial $p(\cdot)$, or satisfying $\ezk + \sqrt{\es} < 1$ in the case where $\ezk$ and $\es$ are constants and $\ec = o(1)$. However, the regime where $\ezk + \es < 1$ but $\ezk + \sqrt{\es} \geq 1$ remained open until this work.

\iffalse{
 {\em interactive} \ZK in constant rounds.

\begin{theorem}[Informal]
    If $\NP \not\subseteq \ioPpoly$ and non-trivial public-coin, constant round (computational) $\ZK$ arguments for $\NP$ exist, then one-way functions exist.
\end{theorem}

The above result requires the assumption that the underlying ZK argument is public-coin. We also obtain a result for private-coin protocols by showing that unless auxiliary-input one-way functions exist, we can convert any constant-round private-coin protocol into a public-coin one. Thus, non-trivial constant round \ZK protocols yield auxiliary-input \OWFs.

}\fi

\paragraph{Application: Amplifying \ZK.} 
Arguably one of the most natural questions about non-trivial \ZK proofs is whether they can be {\em amplified} to drive down their errors to negligible, and obtain ``standard'' \ZK. For $\NIZKs$, this question has been studied extensively recently, and such amplification has been shown to be possible under additional assumptions such as the existence of public-key encryption~\cite{GJS19}, and even just one-way functions~\cite{BG,AK25}.

Our theorem above can be viewed as eliminating the need to make any cryptographic assumptions when amplifying \NIZK. In more detail, if $\NP$ is easy in the worst case then \NIZKs for $\NP$ exist unconditionally, whereas if $\mathsf{NP}$ is hard in the worst case, then Informal Theorem \ref{thm:inf1} builds one-way functions from non-trivial \NIZKs for $\NP$. The resulting one-way functions can be used, together with weak \NIZKs, to obtain standard \NIZKs by relying on known amplification theorems in~\cite{BG,AK25}. 

For technical reasons we will still need to make the extremely mild complexity-theoretic assumption (Assumption 1) that {\em 
$\NP \not\subseteq (\ioPpoly)\setminus\BPP$}. In other words, if $\NP$ languages are efficiently decidable infinitely often with advice, then they are efficiently decidable on all input lengths.  This is an extremely mild complexity assumption (it is implied by $\NP \not \subseteq \ioPpoly$, which is already a  mild worst case assumption).

\begin{corollary}[Informal]
    Under Assumption 1, non-trivial (computational) \NIZK proofs imply standard \NIZK proofs.
\end{corollary}

Note that the corollary above applies only to \NIZK proofs (with statistical soundness), since we currently do not have approaches to amplify arbitrary non-trivial 
%(BG can amplify arguments with negligible $\es$)} 
\NIZK arguments from assumptions weaker than public-key encryption.

A natural next question is whether similar amplification is possible for {\em interactive} \ZK. Here, we note that one-way functions imply the existence of four-round \ZK arguments~\cite{BJY97}, yielding the following informal Corollary.
\begin{corollary}[Informal]
    Under Assumption 1, non-trivial public-coin constant-round interactive (computational) \ZK arguments imply four-round \ZK arguments.
\end{corollary}

\section{Technical Overview}
For simplicity throughout this overview, we will highlight our techniques assuming that the completeness error $\ec = 0$. Formally we work with the general case where $\ec \neq 0$, but this is a straightforward extension of the analysis described in this overview.

In the following two subsections, we recall techniques used in prior works \cite{OstWig,CHK} to show that zero-knowledge requires one-way functions, as well as barriers that prevent these techniques from working when $\es + \ezk \approx 1$. These will build intuition towards our eventual approach. Readers that are already familiar with these works may skip ahead to our techniques described in Section \ref{sec:newover}.

We discuss prior ideas in the context of \emph{non-interactive} zero-knowledge (\NIZK), and then describe methods to close the gap in parameters in the non-interactive setting. However, we note that the more general interactive case requires the most significant new insights, discussed in Section \ref{sec:izk}. 

\subsection{Previous Techniques}

%Prior techniques and barriers encountered for \NIZK will also help shed light on (even more significant) barriers that arise when considering interactive zero-knowledge.

\paragraph{\cite{OstWig} Techniques.}
The key tool used in \cite{OstWig} is Universal Extrapolation (UE): informally, UE lets us ``reverse sample'' the randomness that would generate a particular output of an efficient process.
In more detail, given a polynomial-time sampler $\Samp$ and an output $y$ from its distribution, UE refers to the ability to efficiently and uniformly sample (up to small error) from the set of all inputs that 
$\Samp$ maps to $y$.
UE is possible if and only if one-way functions do not exist~\cite{ImpLev}.
%reference string \crs, yielding an (approximately) uniform $r$ such that 
%$\Sim(x;r)$ outputs the same \crs. 

%which intuitively says that if \OWFs do not exist, then given the description of any polynomial-time algorithm $\Samp$, we can efficiently ``reverse sample'' the randomness that an to generate a given string

%roughly,  a polynomial-time sampler $\mathsf{Samp}$  it is possible to efficiently sample uniformly (up to arbitrarily small inverse-polynomial error) from the distribution of all inputs on which $\mathsf{Samp}$ outputs $y$.  Informally, UE allows us to  \CRS, getting (approximately) a random $r$ such that $\Sim(x;r)$ would output the same $\crs$.

%then there is an {\em efficient} algorithm to sample (nearly) uniformly from the set of input randomness that would lead a given randomized algorithm to generate a given output.  Using UE (which from \cite{ImpLev} we know exists exactly when one-way functions do not), 
Towards a contradiction,~\cite{OstWig} assume the non-existence of one-way functions. Then using UE, they convert a \NIZK $(\Gen, \Pro, \Ver)$ and simulator $\Sim$ for a language $\Lang$ into an algorithm $\Alg$ that decides $\Lang$. In particular, $\Alg$ samples a \emph{real} \CRS from $\Gen$, then uses UE to find a random string $r$ that would lead the \emph{simulator} to output that same \CRS.  $\Alg$ then accepts if and only if $\Ver$ would accept the proof output by $\Sim(x;r)$.

To argue that $\Alg$ works, one has to show that it will accept $x \in \Lang$ with higher probability than it accepts $x \not \in \Lang$.  Consider each of these cases:
\begin{itemize}
    \item If $x \not \in \Lang$, $\Alg$ can only accept if it finds a proof \emph{relative to a real \CRS} that causes $\Ver$ to accept $x$.  Soundness of the underlying \NIZK will immediately tell us that this can happen with probability at most $\es$.

    \item If $x \in \Lang$, note that by completeness $\Ver$ will accept a real \CRS and proof with probability 1
    %\footnote{Formally one can in fact work with an arbitrary completeness error $\ec$, but for simplicity throughout the overview we will assume $\ec = 0$.}, 
    and so by zero-knowledge will accept a simulated \CRS and proof with probability at least $1 - \ezk$.

    Now consider an algorithm $\cB$ that is \emph{almost} identical to $\Alg$, except that it starts by sampling a \CRS from $\Sim$ instead of from $\Gen$.  We note first that this is (up to the small statistical error of the UE) identically distributed to simply running $\Ver$ on a simulated \CRS and proof, so by the above $\cB$ accepts $x \in \Lang$ with probability at least $1 - \ezk$.  But note that since the only difference between $\Alg$ and $\cB$ is in the original sampling of the \CRS, and so by zero-knowledge their outputs must be $\ezk$-close.  Thus, $\Alg$ accepts $x \in \Lang$ with probability at least $1 - 2\ezk$.
\end{itemize}

Thus, $\Alg$ is a correct algorithm for $\Lang$ as long as $1 - 2\ezk$ is noticeably larger than $\es$, which is equivalent to requiring that $\es + 2\ezk$ is noticeably smaller than 1.  Since one can implement $\Alg$ as long as one-way functions \emph{don't} exist, this means that a \NIZK for a hard language with $\es + 2\ezk < 1$ \emph{requires} the existence of a one-way function.

Note that in the analysis above, we first uses the zero-knowledge property to bound the difference between real and simulated proofs (costing one $\ezk$) and then again bound the difference between real and simulated $\crs$ (costing another $\ezk$). It is natural to wonder whether there may be a better way to analyze the same algorithm and show that $\cA$ decides $\Lang$ whenever we have arbitrary {\em non-trivial \NIZKs}, but in the following subsection (Section \ref{sec:barrier}) we provide an explicit counterexample.

\paragraph{\cite{CHK} Techniques.}

The improvement in~\cite{CHK} stems from the observation that one of the $\ezk$ errors occurs specifically when the analysis transitions the $\crs$ (and not the proof itself) from real to simulated. They then add a ``CRS check'' to the deciding algorithm $\cA$ described above, which roughly predicts whether a \CRS was sampled from the real or simulated distribution.  This requires the use of Universal Approximation (UA), which intuitively allows one to approximate the probability of a given output coming from a given randomized process. Similar to the case of UE, \cite{ImpThesis} showed that UA is possible exactly when one-way functions don't exist.\footnote{\cite{CHK} technically uses a stronger notion of UA than the one considered in \cite{ImpThesis}, but for our purposes we will ignore this subtlety.}  

In particular, \cite{CHK} modifies the algorithm $\Alg$ above by adding a step that estimates the probability that the \CRS it samples came from $\Gen$ or from $\Sim$, and immediately rejects if the latter is at least $\frac{1}{\sqrt{\es}}$ times larger than the former.
As in \cite{OstWig}, one now has to show that the modified $\Alg$ accepts $x \in \Lang$ more often than $x \not \in \Lang$.
\begin{itemize}
    \item If $x \not \in \Lang$, the above modification can only make $\Alg$ less likely to accept, so the acceptance probability is still upper-bounded by $\es$.
    \item If $x \in \Lang$, all the hybrids considered above have to be modified to also include the \CRS checking step.  It is possible to show that a real \CRS fails the check with probability at most $\sqrt{\es}$ (ignoring a small error term from the UA),
    and thus a real \CRS and proof will be accepted \emph{and} pass the \CRS check with probability at least $1 - \sqrt{\es}$.  Thus, by zero-knowledge a \emph{simulated} \CRS and proof must pass both checks with probability at least $1 - \sqrt{\es} - \ezk$.

    The key step in the \cite{CHK} analysis comes when considering $\cB$, which as before is equivalent to $\cA$ except that it starts with a simulated \CRS.  As in the prior analysis, this is identically distributed (up to the small statistical error of the UE) to sampling the \CRS and proof from $\Sim$ directly, and so will accept $x \in \Lang$ with probability at least $1 - \sqrt{\es} - \ezk$.  However, instead of paying an \emph{additive} $\ezk$ loss in acceptance probability between $\cB$ and $\Alg$, \cite{CHK} notes that one can instead pay a \emph{multiplicative} $\sqrt{\es}$ loss.  Indeed, any \CRS that is at most $\frac{1}{\sqrt{\es}}$ times more likely to come from $\Sim$ as from $\Gen$ can ``contribute'' at most $\frac{1}{\sqrt{\es}}$ times more the the acceptance probability of $\cB$ as it does to $\cA$; all other \CRS are rejected by the  checking step, and so ``contribute'' nothing in either algorithm.  Thus, $\cA$ accepts $x \in \Lang$ with probability at least $\sqrt{\es}$ times the acceptance probability of $\cB$, which gives a final bound of $\sqrt{\es}\left( 1 - \sqrt{\es} - \ezk \right)$.
\end{itemize}

This means that $\Alg$ is a correct algorithm for $\Lang$ as long as $\sqrt{\es}\left( 1 - \sqrt{\es} - \ezk \right)$ is noticeably smaller than $\es$; rearranging terms, this requires that $\ezk + 2\sqrt{\es}$ is noticeably smaller than 1.  As before, this says that a \NIZK for a hard language with $\ezk + 2 \sqrt{\es} < 1$ must imply the existence of a one-way function. 

\subsection{Barriers to Improvement}
\label{sec:barrier}

We would ideally like a result that works for \emph{any} non-trivial \NIZK, meaning it should only need to assume that $\es + \ezk < 1$.  One might hope that a tighter analysis of the hybrids in \cite{OstWig} or \cite{CHK} could achieve such a result, but this turns out to be impossible: there exists a (contrived) protocol with $\es + \ezk < 1$ for which neither algorithm $\Alg$ described above give advantage in distinguishing $x \in \Lang$ from $x \not \in \Lang$!

In particular, consider some language $\Lang \in \classP$\footnote{Note that while this language is by definition easy, that doesn't matter for the purposes of our counterexample; we are simply trying to show that \emph{the particular algorithm} $\Alg$ doesn't decide $\Lang$.} and define a \NIZK for $\Lang$ with errors $(\es, \ezk)$ as follows:
\begin{itemize}
    \item $\Gen:$ Output 
    \[\crs = 
        \begin{cases}
        0 & \text{w.p. } 1-\ezk, \\
        1 & \text{w.p. } \ezk.
        \end{cases}\]
    \item $\Pro (\crs)$ outputs as proof a bit $\pi = \crs$.
    \item $\Ver$ decides $x$ on its own and if $x \in \Lang$, $\mathsf{V}$ outputs
    \[\mathsf{V}(\crs,\pi,x) = 
        \begin{cases}
        1 & \text{if }  \pi = \crs, \\
        0 & \text{if }  \pi \neq \crs. 
        \end{cases}\]
    If $x \not \in \Lang$, $\mathsf{V}$ outputs
    \[\mathsf{V}(\crs,\pi,x) = 
        \begin{cases}
        1 & \text{w.p. }  \es, \\
        0 & \text{w.p. }  1-\es.
        \end{cases}\]
    \item $\Sim(x)$ outputs \[(\crs, \pi) = 
    \begin{cases}
        00 &  \text{w.p. }1-\ezk,\\ 
        01 &  \text{w.p. }\ezk.
        \end{cases}\]
%        probability $1 - \ezk$ and $01$ with probability $\ezk$.
    %- \delta$, and $10$ with probability $\delta$, where $\delta$ is taken to be inverse exponentially small (say $2^{-|x|})$.\footnote{For simplicity in the following we will treat $\delta$ as 0; the only reason to have $\delta$ at all is so that the behavior of $\Sim$ conditioned on the \CRS being 1 is well-defined.} 
\end{itemize}

One can easily verify that the soundness error is exactly $\es$ and the zero-knowledge error is exactly $\ezk$ by construction. 

But now consider the behavior of the algorithm $\Alg$ from \cite{OstWig} if instantiated with this \NIZK (assuming, for simplicity, that the UE has no errors):
%\footnote{Below, we assume that the UE has no errors; allowing (small but adversarial) errors in the UE can only make things worse.} :
\begin{itemize}
    \item If $x \not \in \Lang$, $\Ver$ will accept with probability $\es$ regardless of what \CRS and proof it is given, so $\Alg$ will accept with the same probability.
    \item If $x \in \Lang$, note that $\Alg$ can never sample the proof to be $1$ if the \CRS is $1$, since $\Sim$ never outputs $11$.\footnote{Technically, $\Sim$ can never output a CRS of 1, and so the behavior of $\Alg$ if it samples 1 from $\Gen$ is undefined.  Informally, we can say that if the distribution of $\Sim$ is undefined, $\Alg$ simply fails.  Alternatively, we can fix this formally by having $\Sim$ output $10$ with some tiny probability $\delta$, where we set $\delta$ to be small enough to not affect our analysis meaningfully.}  Thus, the only way for $\Alg$ to accept is if it samples the \CRS as $0$ (which happens with probability $1 - \ezk$ from $\Gen$) and the proof as $0$ (which happens with probability $(1 - \ezk)$ from $\Sim$ conditioned on the \CRS being $0$). Thus, $\Alg$ will only accept $x \in \Lang$ with probability $(1 - \ezk)^2$.
\end{itemize}
Thus, at a minimum, we need to assume that $(1 - \ezk)^2$ is noticeably larger than $\es$ in order for the \cite{OstWig} algorithm to have any chance of working.  This rules out the possibility of a better analysis of the same algorithm giving us the optimal result.\footnote{As an explicit example, the parameter setting of $\ezk = 0.5$ and $\es = 0.25$ is non-trivial, but the \cite{OstWig} algorithm would accept $x$ with probability $0.25$ regardless of if $x \in \Lang$ or not.}

Additionally, the same counterexample also rules out the possibility of getting a tight result from the modified decision algorithm $\Alg$ in \cite{CHK}.  Indeed, if the parameters are set such that the \CRS check rejects when the \CRS is 0, $\Alg$ will \emph{never} accept $x \in \Lang$.  If instead the \CRS check accepts when the \CRS is 0, then it accepts \emph{every} \CRS\footnote{Note that $\Sim$ outputs the \CRS as $1$ with probability \emph{smaller} than $\Gen$ does, so the \CRS check can never reject this.}, and thus (for this particular \NIZK) the algorithm from \cite{CHK} is functionally equivalent to that from \cite{OstWig}!

\subsection{Repetition to the Rescue}
\label{sec:newover}

In order to improve on prior techniques, it will help to first understand \emph{why} the above counterexample is bad for them.  We note that the counterexample hinges on the fact that on $\crs = 0$ the simulator always has the \emph{ability} to output a \good proof (i.e. $0$), but still chooses to output a \bad proof (i.e. $1$) some of the time.
Here by \good we mean a proof that is statistically close to a real proof, and by \bad we mean a proof that is far from a real proof. This causes the distribution used by $\Alg$ to have $2\ezk$ error (once due to the distance between the real and simulated \CRS and once due to the distance between the real and simulated proofs \emph{conditioned on} a particular \CRS), while still keeping the real and simulated (\CRS, proof) distributions $\ezk$-close overall.

It turns out that this type of simulator behavior is the \emph{only} reason prior techniques don't work with protocols where $\es + \ezk \approx 1$.  In particular, for every fixing of $x \in \Lang$, if every \CRS was either \good (i.e., the simulator conditioned on this \CRS \emph{always} outputs an accepting proof for $x$) or \bad (i.e., the simulator conditioned on this \CRS \emph{never} outputs an accepting proof for  $x$), then the algorithm $\Alg$ from \cite{OstWig} {\em will work} only assuming that $\es + \ezk < 1$!  This is because on any \bad \CRS output by $\Gen$, the real prover is guaranteed to output an accepting proof whereas the simulator never does. This allows us to infer that $\Gen$ must output a \good \CRS with probability at least $1 - \ezk$, and thus $\Alg$ will accept $x \in \Lang$ with the same probability.

Of course, we cannot in general assume that the simulator is guaranteed to always output \good proofs on certain {\CRS}es and always output \bad proofs on others. But, we can hope to mimic a similar effect by sampling \emph{many} simulated proofs for the same \CRS and checking whether {\em they all} reject. Our updated definition of a $\good$ \CRS is one where (with overwhelming probability) {\em at least one} of the simulated proofs that is sampled is accepted.\footnote{We note that \cite{Ps05} use a similar idea, where they categorize the $\crs$ according to whether there \emph{exists} an accepting proof relative to it.  However, their analysis is limited to the negligible-error setting and doesn't seem to naturally generalize to the interactive setting. }
%\footnote{We note that \cite{Ps05} also categorize the $\crs$ according to whether there \emph{exists} an accepting proof relative to it, however their analysis is limited to the negligible-error setting. (We thank an anonymous STOC reviewer for pointing this out.) On the other hand, our presentation here more easily generalizes to the case of arguments as well as the interactive setting.}
%the categorization in this way would limit us to the case of statistical soundness, as well as not generalizing as naturally to the interactive case.}

In particular, we fix an appropriate polynomial $p(\cdot)$ and consider a modified algorithm $\Alg'$ that repeats the %UE step from the 
\cite{OstWig} algorithm $p(|x|)$ times (all with respect to the same \CRS, but reverse sampling fresh randomness each time) and accepts if \emph{any} of the repetitions give an accepting proof.
\begin{itemize}
    \item For $x \not \in \Lang$, $\Alg'$ still needs to find an accepting proof relative to a real \CRS in order to accept.  Since repeatedly running the~\cite{OstWig} algorithm is a valid strategy to break soundness of the \NIZK, we have that the probability that even one of the proofs accept is at most $\es$.
    \item For $x \in \Lang$, we call a \CRS \good if $\Sim$ conditioned on outputting \CRS gives an accepting proof for $x$ with probability $\geq \frac{1}{p(|x|)}$ and \bad otherwise.
    %\footnote{In the formal proof, we will actually consider the probability that \emph{UE} run on the simulator finds an accepting proof.  However, for the sake of simplicity in the overview, we will phrase this as the true conditional distribution of $\Sim$'s output instead.}. 
    Very roughly, as long as $p(|x|)$ is large enough, we can as above argue that $\Gen$ can only output a \bad \CRS with probability (approximately) $\ezk$, since doing so contributes directly to the zero-knowledge error of the \NIZK\footnote{In more detail note that $\Sim$ will almost never output a \bad \CRS with an accepting proof, since by definition the probability $\Sim$'s proof is accepting conditioned on the \CRS being bad is at most $\frac{1}{p(|x|)}$.  Since we are assuming perfect completeness however, $\Pro$ will always output an accepting proof, even when the \CRS from $\Gen$ is \bad.  Thus, if $\Gen$ outputs a \bad \CRS with probability much higher than $\ezk$, we can break zero-knowledge with a distinguisher that checks if the \CRS is \bad but the proof is still accepting. Here, note that we can (probably approximately) estimate whether a \CRS is \good or \bad \emph{efficiently} by using UE to sample many possible simulated proofs relative to it, and checking what fraction of them are accepting.}.
    And if $\Gen$ outputs a \good \CRS, $\Alg'$ will find an accepting proof for $x$ with overwhelming probability.  Thus, up to small errors, $\Alg'$ will accept $x \in \Lang$ with probability at least $1 - \ezk$. 
\end{itemize}

As in prior analysis, we just need the probability that $\Alg'$ accepts $x \in \Lang$ to be noticeably larger than the probability it accepts $x \not \in \Lang$.  Rearranging terms, this only requires that $\es + \ezk$ be noticeably smaller than 1---meaning we can work with \emph{any} non-trivial \NIZK!

We have swept some technical details under the rug, but the above ideas coarsely suffice to obtain \OWFs from non-trivial \NIZKs.
In the following section, we discuss barriers and methods to extend these ideas to the interactive setting.

\subsection{Generalizing to Interactive Zero Knowledge}
\label{sec:izk}
%While the above fully resolves the relationship between \NIZKs and one-way functions, we would ideally like to have a matching result for interactive zero-knowledge.  
As discussed in the introduction, our primary goal is to obtain an implication from non-trivial {\em honest-verifier} \ZK to \OWFs, since this gives us the strongest possible theorem statement.
First, we note that techniques from \cite{OstWig} naturally extend to the interactive case, but get further from the desired result as the number of rounds increases: to get a one-way function from a $k$ round protocol, they require $\es + k\ezk < 1$. Similarly, simple extensions of the~\cite{CHK} techniques require $(\es)^{1/k} + \ezk < 1$.

A natural goal is therefore to see if we can extend the techniques that gave us a tight result for \NIZKs to the interactive setting, and more specifically, eliminate the dependence on $k$ in the above constraints.

\iffalse{
At first, it might seem as though our new techniques can't be applied to interactive zero-knowledge at all, since they require us to check if a prover message leads to an accepting versus a rejecting proof.  In the interactive case, this check is ill-defined as the verifier only decides whether to accept at the very end of the interaction; if we wait to apply our techniques to the very last prover message where this isn't a problem, it may be that \emph{no} message can change the outcome!\footnote{In particular, given any interactive protocol one can always construct a related protocol that simply adds a ``dummy'' round at the end.  Nothing we do with just the last round of such a modified protocol can possibly help us in any way.}

In order to generalize our techniques to the interactive setting, it will help to first rephrase slightly what we are doing in the non-interactive case.  Instead of viewing our algorithm $\Alg$ as generating many possible proof strings $\prf$ and checking if any of them would be accepted (which would be ill-defined if we tried to apply it part way through an interactive protocol), we can instead treat $\Alg$ as creating many possible proof strings $\prf$ and keeping only the one that maximizes the probability that $\Ver$ accepts.  Since that probability is always either 0 or 1 (corresponding to $\Ver$ rejecting or accepting $\prf$), this new formulation is equivalent to the original!
}\fi

In order to generalize our techniques to the interactive setting, it will help to take a slightly different view of our proof in the non-interactive case.  We can treat our algorithm $\Alg$ as defining a malicious prover $\widetilde{\Pro}$, then simulating a run of $\widetilde{\Pro}$ and $\Ver$ to decide whether to accept $x$.  In the analysis above, we treated $\widetilde{\Pro}$ as using UE to sample many possible simulated proof strings, and outputting an accepting one if found.  But note that we equally well could phrase $\widetilde{\Pro}$ as choosing whichever sampled proof string has the maximum probability of being accepted. Phrased this way, we can now begin to extend this idea to the interactive case.  %

\paragraph{An Inefficient Prover.}
We would like to define a malicious prover $\widetilde{\Pro}$ such that our algorithm $\cA$ to decide $\Lang$ can simply simulate an interaction between $\widetilde{\Pro}$ and $\Ver$ on $x$ and output the result. Unlike the non-interactive case, $\widetilde{\Pro}$ can't immediately check if the message it is considering would lead to $\Ver$ accepting or rejecting, since this decision may depend on randomness sampled at a later point in the protocol.  Thus, all we can hope for is to have $\widetilde{\Pro}$ choose the ``best'' message available to it, which means choosing the message that maximizes the probability that $\Ver$ accepts \emph{if we continue to apply the prover strategy $\widetilde{P}$ in future rounds}. For now, we will assume that $\widetilde{P}$ can (inefficiently) find such messages, and we will later see how to make $\widetilde{P}$ efficient.

More formally, suppose for simplicity that we are starting from a protocol that satisfies statistical soundness and zero-knowledge.  We can (recursively) define $\widetilde{\Pro}$ that, given $x$ and a partial transcript $\transcript$, considers all possible next messages $m$ in the support of $\Sim$
%\footnote{Technically, this will be ill-defined if no full transcript in the support of $\Sim$ begins with $\transcript$.  In the formal proof, we get around this issue by considering the support of the UE \emph{applied to} $\Sim$.  For simplicity in the overview though we will ignore this detail.} 
and outputs whichever one maximizes the probability of $\Ver$ accepting $x$ if it interacts with $\widetilde{\Pro}$ starting from transcript $\transcript \| m$. Here, defining this success probability based \emph{only} on the transcript and not on $\Ver$'s internal state limits us to only working with public-coin protocols.

Let $\Alg(x)$ be the algorithm that runs $\widetilde{\Pro}$ and $\Ver$ on input $x$ and outputs $\Ver$'s decision bit; we can now argue that $\Alg$ correctly decides $\Lang$.
\begin{itemize}
    \item If $x \not \in \Lang$, statistical soundness tells us that no prover can convince $\Ver$ with probability greater than $\es$, and so in particular $\widetilde{\Pro}$ cannot.  Thus, $\Alg(x)$ will accept with probability at most $\es$.

%    \item If $x \in \Lang$, we can argue that $\widetilde{\Pro}$ performs ``almost as well as'' the honest prover $\Pro$.  In particular, we first note that the probability that $\Pro$ outputs \emph{any} message outside the support of $\Sim$ is bounded by $\ezk$ from statistical zero-knowledge.\footnote{Importantly, note that here we are applying zero-knowledge \emph{globally} to the entire transcript instead of on a message-by-message basis.  This is the key reason why we avoid the $O(k)\ezk$ term that the analysis from \cite{OstWig} has.}  Additionally, since $\widetilde{\Pro}$ always chooses the \emph{best} message in the support of $\Sim$, we know that the probability that $\widetilde{\Pro}$ convinces $\Ver$ to accept is at least the probability that $\Pro$ convinces $\Ver$ \emph{conditioned on $\Pro$ always choosing messages from the support of $\Sim$}.  Combining these two facts, we get that the probability that $\widetilde{\Pro}$ convinces $\Ver$ can be no more than $\ezk$ smaller than the probability that $\Pro$ convinces $\Ver$. 
%    Since in the overview we are assuming perfect completeness, this tells us that $\Alg(x)$ will accept with probability at least $1 - \ezk$. 

    \item If $x \in \Lang$, we can argue that $\widetilde{\Pro}$ performs ``almost as well as'' the honest prover $\Pro$.  In particular, we note that since $\widetilde{\Pro}$ always chooses the best possible message in the support of $\Sim$, the only way for $\Pro$ to choose a better message is if it sends a message that $\Sim$ never would.  But from statistical zero-knowledge, we know that this can happen with probability at most $\ezk$!\footnote{Importantly, we are applying zero-knowledge \emph{globally} to the entire transcript instead of on a message-by-message basis.  This is the key reason why we avoid the $O(k)\ezk$ term in the analysis from \cite{OstWig}.} Thus, the probability that $\widetilde{\Pro}$ convinces $\mathsf{V}$ is at least $(1-\ezk)$, which tells us that $\cA$ accepts with probability at least $(1-\ezk)$.
    
    %Formally, we have that the probability that $\widetilde{\Pro}$ convinces $\Ver$ is at least the probability that $\Pro$ convinces $\Ver$ \emph{conditioned on the output transcript being in the support of $\Sim$}.  But we know that this conditional probability is at least the probability that $\Pro$ convinces $\Ver$ minus the probability that $\Pro$ and $\Ver$ create a transcript outside the support of $\Sim$; applying (perfect) completeness and statistical zero-knowledge, this is at least $1 - \ezk$.
\end{itemize}

As in the \NIZK case, this immediately gives us an algorithm for $\Lang$ as long as $\es + \ezk$ is noticeably smaller than 1.

\paragraph{Making $\widetilde{\Pro}$ Efficient.}

So far, we required $\widetilde{\Pro}$ to find the \emph{best} possible next prover message in the support of $\Sim$, which may not be efficient to compute.
%\footnote{For example, if $\Sim$ has a small probability of outputting a fully random transcript (so that every message is in its support), this would require $\widetilde{\Pro}$ to implement the optimal prover strategy!}  
This inefficiency meant %we were stuck working with statistically sound protocols---and more importantly, 
that our eventual decision algorithm $\Alg$ is inefficient, which is a problem. 
To fix these issues, we need to modify how we define $\widetilde{\Pro}$ so that it has a polynomial run time, but still works well enough to argue that $\Alg(x)$ will accept with good probability when $x \in \Lang$.

Taking inspiration from the \NIZK case, we can have $\widetilde{\Pro}$ use UE to sample many possible next messages from $\Sim$, and simply choose the best message \emph{among those it sampled}.  While this won't guarantee that we get the best \emph{possible} next message, it does guarantee that (with overwhelming probability) we get a message that is ``one of the best''.  In particular, suppose we order all the messages in the support of $\Sim$ by the probability that $\widetilde{\Pro}$ will convince $\Ver$ if it chooses that message next.  For any $\delta$, if we have $\widetilde{\Pro}$ sample $\frac{n}{\delta}$ many possible next messages, it will sample at least one in the top $\delta$ fraction of this order with overwhelming probability. 
%\footnote{Formally, the probability that it \emph{never} samples a message in the top $\delta$ fraction is $(1 - \delta)^{n/\delta} \leq e^{-n}$.} 
\iffalse{The above analysis is robust enough to handle this approximation error as long as $\delta$ is a sufficiently small inverse polynomial in $n$.\footnote{We do have to modify the distinguisher we use slightly so that instead of checking membership in the \emph{entire} support of $\Sim$, it checks if every prover message is in the bottom $(1 - \delta)$ fraction of the support.  Note that for a transcript from $\Sim$, this will hold with probability all but $\delta k$ by a union bound.} }\fi

There is however one catch to this strategy that only arises in the interactive case: in order to choose the best message among those sampled, $\widetilde{\Pro}$ has to be able to compute the probability that each message will lead to it convincing $\Ver$.  Note though that as long as $\Ver$ is public coin (and so stateless), $\widetilde{\Pro}$ can get a very good estimate of this probability, as $\widetilde{\Pro}$ can simulate an interaction between itself and $\Ver$ ``in its own head''.  A Chernoff bound says that $\widetilde{\Pro}$ only needs to run $\poly(n, \frac{1}{\delta})$ such simulations in order to (with overwhelming probability) estimate its success probability to within an additive $\delta$ factor for any $\delta$. Note, however, that since $\widetilde{\Pro}$ has to simulate itself a polynomial number of times on a transcript that is only one message longer, its run time will grow by a \emph{multiplicative} factor of $\poly(n)$ for each round.  This will give us a final run time of $\poly(n)^k$, which constrains us to only work with protocols where $k$ is a constant.

\iffalse{If we plug this estimation procedure into $\widetilde{\Pro}$ as described above, we have that with overwhelming probability, it always chooses a prover message whose success probability is within $2\delta$ of a message in the top $\delta$ fraction of the support of $\Sim$.  Our analysis is still robust even to this form of approximation error as long as $\delta$ is chosen to be a sufficiently small inverse polynomial.} \fi

Putting this all together, we get an \emph{efficient} $\widetilde{\Pro}$ that with overwhelming probability chooses a next message that has success probability at most $2\delta$ smaller than some message in the top $\delta$ fraction of messages in the support of $\Sim$.  We claim that as long as $\delta$ is set to a sufficiently small inverse polynomial, our prior analysis is robust enough to work with these weaker guarantees on $\widetilde{\Pro}$.

In more detail, note that all we have to change in the above analysis is the argument that $\Pro$ can't pick a ``better'' message than $\widetilde{\Pro}$ too commonly.  Since our weaker guarantee on $\widetilde{\Pro}$ is that it picks a message (approximately) in the top $\delta$ fraction of possible messages from the support of $\Sim$, we just have to bound the probability that $\Pro$ picks a message that is better than a $(1 - \delta)$ fraction of the support of $\Sim$. 

By a union bound over each round, the probability that $\Sim$ ever outputs such a message is at most $\delta k$, so by (statistical) zero-knowledge, the probability that a real transcript contains such a message is at most $\ezk + \delta k$.  By our guarantees on $\widetilde{\Pro}$, we know that it loses at most $2\delta$ success probability in each round compared to $\Pro$ conditioned on not picking such a ``too good'' message.  Thus, since $\Pro$ will always convince $\Ver$ to accept $x \in \Lang$, we have that $\widetilde{\Pro}$ convinces $\Ver$ with probability at least $1 - \ezk - 3\delta k$. 

Finally, in order for $\Alg$ to be correct, we will need $\es$ to be noticeably smaller than $1 - \ezk - 3\delta k$.  As long as $\es + \ezk$ is noticeably smaller than 1, we can pick an appropriate inverse-polynomial $\delta$ that makes this happen.

\subsubsection{Working with Computational Zero-Knowledge}

Up until this point, we've been applying zero-knowledge with respect to a distinguisher $\cD$ that \emph{perfectly} answers questions about the support of $\Sim$---in particular, one that checks if a given prover message better than the bottom $(1 - \delta)$ fraction of the support of $\Sim$. 
Here recall that a message $a$ is better than $b$ if the probability that $\widetilde{P}$ convinces $\mathsf{V}$ is higher if $\widetilde{P}$ chooses $a$ in its next round, compared with $b$. 

If we want to work with protocols that only have computational zero-knowledge, we need to come up with an efficient distinguisher that still works well enough for our proof.

Our efficient distinguisher will try to {\em approximate} the fraction of messages from $\Sim$ that a given message is better than.
To do this, it will sample a large number of possible next messages from the support of $\Sim$ and compare them to the given message.  Formally, it will sample $\poly(n, \frac{1}{\delta})$ such messages, and check if at most a $2\delta$ fraction of them are better than the given prover message.
By a Chernoff bound, we have that if the given prover message is in fact better than a $(1 - \delta)$ fraction of all possible messages in the support of $\Sim$, this check will pass with overwhelming probability.

As in the previous section, we note that while we may not be able to compute the \emph{exact} success probability efficiently, the same estimation procedure can with overwhelming probability get us within an additive factor of $\delta$ for any inverse-polynomial $\delta$.  Thus, except with negligible probability, we know that our distinguisher will say ``real'' if \emph{any} prover message in its input transcript is at least $2\delta$ better than the bottom $(1 - \delta)$ fraction of messages in the support of $\Sim$. 

By a careful analysis, we can show that 
%We now claim 
that this efficient distinguisher still works for the analysis in the previous section. 
Indeed, there are only two points where we had to use properties of the distinguisher: first to argue that it rarely outputs ``real'' on a simulated transcript, and second to argue that $\widetilde{\Pro}$ does ``almost as well'' as $\Pro$ conditioned on the distinguisher saying ``simulated''.
\begin{itemize}
    \item For the former case, we claim that the efficient distinguisher says ``real'' on a transcript from $\Sim$ with probability at most $2\delta k$.  For any given round, we know that the prover message in the transcript was generated by $\Sim$, and hence comes from the same distribution as the messages the distinguisher samples.  Thus, the probability that the message in the transcript is in the top $2\delta$ fraction of the sampled messages is just $2\delta$, regardless of what procedure we use to compare them.

    \item For the latter case, we note that the only change between this section and the last is that the efficient distinguisher may allow $\Pro$ to have messages that have up to $2\delta$ higher success probability than the bottom $(1 - \delta)$ fraction of messages in the support of $\Sim$.  Since in the last section we already allowed $\widetilde{\Pro}$ to output messages that have up to $2\delta$ \emph{lower} probability than the top $\delta$ fraction of messages, we essentially are just doubling the error from there. 
\end{itemize}
In both cases, we see that setting $\delta$ to be half as large as the setting of $\delta$ in the previous section suffices to make the same analysis go through. 

\subsubsection{Working with Private Coins}

The above proof critically relies on the zero-knowledge protocol we start with being public-coin so that we can efficiently simulate $\Ver$ starting from an arbitrary transcript and thus estimate the probability of it accepting.  However, we note that at least intuitively, we shouldn't have to make this assumption: if $\Ver$ is able to computationally hide information about its randomness in a private-coin protocol, intuitively the map from $\Ver$'s randomness to its messages is a one-way function.

We formalize this intuition in Appendix \ref{sec:private-to-public}, showing that either infinitely-often auxiliary-input uniform-secure one-way functions exist or every constant-round private-coin protocol can be converted into an equivalent public-coin protocol. 

We note that, unlike in the case of public-coin protocols, we obtain \textit{auxiliary-input} one-way functions. This is because in our transformation, in order to achieve soundness the verifier has to resample its randomness such that it is consistent with the full transcript. Since the verifier does not have access to the malicious prover's code, there does not immediately seem to be a way to construct an efficient sampler such that UE on that sampler will allow us to do this.  We get past this by allowing the sampler to take the transcript as auxiliary input, but this limits our result to achieving auxiliary-input one-way functions. Since this auxiliary input comes from the prover messages and not from the worst-case instance, switching to average-case hardness does not seem to improve the situation. Nonetheless, our results are sufficient to rule out information-theoretic constructions of (constant-round) non-trivial ZK even in the private-coin setting.

\section{Preliminaries}

\subsection{Notation}

We define the following notation:
\begin{itemize}
    \item $[n]$ represents the set $\{ 1, 2, \ldots, n \}$ and $[n]_0$ represents the set $\{ 0, 1, \ldots, n \}$.
    \item We call a function $\mu$ negligible if it is smaller than any inverse polynomial, that is if for every polynomial $p$ there exists an $n_0 \in \mathbb{N}$ such that for all $n \geq n_0$, $\mu(n) < \frac{1}{p(n)}$.  We let $\negl$ denote an unspecified negligible function.
    \item We call a function $\mu$ noticeable if it is asymptotically larger than some inverse polynomial, that is if there exists a polynomial $p$ and $n_0 \in \mathbb{N}$ such that for all $n \geq n_0$, $\mu(n) \geq \frac{1}{p(n)}$.
    \item $a <_n b$ represents that $a$ is noticeably less than $b$, that is that there exists a polynomial $p$ and $n_0 \in \mathbb{N}$ such that for all $n \geq n_0$, $a(n) < b(n) - \frac{1}{p(n)}$.
\end{itemize}

\subsection{Uniform and Non-Uniform Machines}
Throughout this paper, our protocols will have (honest) participants that can be described as uniform Turing machines, whereas adversaries will (w.l.o.g.) be non-uniform Turing Machines, or equivalently, families of polynomial-sized circuits.

\subsection{One-Way Functions}

\begin{definition} [One-way Functions] \label{def:owf}
    Fix an alphabet $\auxAlph$.  We say that a family of functions $\{ f_x \}_{x \in \auxAlph^*}$ with $f_x : \zo^{m_i(|x|)} \to \zo^{m_o(|x|)}$ is one-way if
    \begin{itemize}
        \item \textbf{[Efficiency]} There exists a Turing machine $M$ and polynomial $p$ such that for all $x \in \Sigma^*$ and $r \in \zo^{m_i(|x|)}$, $M(x, r)$ outputs $f_x(r)$ in time at most $p(|x|)$
        
        \item \textbf{[One-Wayness]} For all non-uniform PPT machines $\Alg$, all polynomials $p$, and all sufficiently large $n$, there exists an $x \in \Sigma^n$ with
        \begin{equation*}
            \P[f_x(\Alg(x, f_x(r))) = f_x(r) | r \samp \zo^{m_i(|x|)}] \leq \frac{1}{p(n)}
        \end{equation*}
    \end{itemize}
\end{definition}

\begin{definition} [Weak One-way Functions] \label{def:wowf}
    Fix an alphabet $\auxAlph$.  We say that a family of functions $\{ f_x \}_{x \in \auxAlph^*}$ with $f_x : \zo^{m_i(|x|)} \to \zo^{m_o(|x|)}$ is weakly one-way if
    \begin{itemize}
        \item \textbf{[Efficiency]} There exists a Turing machine $M$ and polynomial $p$ such that for all $x \in \Sigma^*$ and $r \in \zo^{m_i(|x|)}$, $M(x, r)$ outputs $f_x(r)$ in time at most $p(|x|)$
        
        \item \textbf{[Weak One-Wayness]} For all non-uniform PPT machines $\Alg$ there exists a polynomial $p$ such that for sufficiently large $n$ there exists an $x \in \Sigma^n$ with
        \begin{equation*}
            \P[f_x(\Alg(x, f_x(r))) = f_x(r) | r \samp \zo^{m_i(|x|)}] \leq 1 - \frac{1}{p(n)}
        \end{equation*}
    \end{itemize}
\end{definition}

\begin{remark}
    A few remarks about these definitions are in order:
    \begin{itemize}
        \item The definition above can be used to define both (standard) one-way functions, and one-way functions with auxiliary input. If we take $\auxAlph = \zo$, these definitions correspond to (weak) one-way functions with auxiliary inputs; if we take $\auxAlph = \{ 1 \}$, these definitions correspond to (weak) one-way functions with no auxiliary input.  In this paper, if we refer to an ``auxiliary input (weak) one-way function'', we mean one with respect to $\auxAlph = \zo$; if we don't otherwise specify $\auxAlph$, we default to $\auxAlph = \{ 1 \}$.  A similar remark applies to all our definitions that make reference to $\auxAlph$.

        \item Since we require that $M$ is computable in time polynomial in $|x|$ (as opposed to in the size of its full input), we can assume WLOG that $m_i$ and $m_o$ are polynomially-bounded functions.
    \end{itemize}
\end{remark}

While weak one-way functions may at first seem much weaker than one-way functions, the two notions are existentially equivalent.

\begin{theorem}[\cite{Yao}]\label{thm:wowf-vs-owf}
    Fix an alphabet $\auxAlph$.  One-way functions exist with respect to $\auxAlph$ if and only if weak one-way functions do.  
\end{theorem}

\subsection{Universal Extrapolation}

\begin{definition}[Sampler]
    Fix an alphabet $\auxAlph$.  We say that a probabilistic polynomial time machine $M$ is a sampler from $m_i$ input bits to $m_o$ output bits if for all $x \in \auxAlph^*$, $M(x)$ uses at most $m_i(|x|)$ bits of randomness and for all $r \in \zo^{m_i(|x|)}$, $M(x; r) \in \zo^{m_o(|x|)}$.
\end{definition}

\begin{remark}
    Note that $m_i$ and $m_o$ can WLOG be bounded by the run time of $M$; since $M$ runs in time polynomial in $|x|$, we can freely assume that $m_i$ and $m_o$ are polynomially-bounded functions of $n$.
\end{remark}

\begin{definition}[Infinitely-Often Universal Extrapolation] \label{def:io-ue}
    Fix an alphabet $\auxAlph$.  We say that Infinitely-Often Universal Extrapolation (io-UE) is possible if for every sampler $M$ from $m_i$ input bits to $m_o$ output bits and all polynomials $p$, there exists a non-uniform PPT machine $N$ such that for infinitely many $n$, the following two distributions are $\frac{1}{p(n)}$-close in statistical distance for every $x \in \auxAlph^n$:
    \begin{enumerate}
        \item Sample $r \samp \zo^{m_i(n)}$ and output $r \| M(x;r)$
        \item Sample $r \samp \zo^{m_i(n)}$ and output $N(x, M(x;r)) \| M(x;r)$
    \end{enumerate}
\end{definition}

\begin{theorem}[\cite{ImpLev}] \label{thm:owf-vs-ioue}
    Fix an alphabet $\auxAlph$.  io-UE is possible with respect to $\auxAlph$ if and only if one-way functions do \emph{not} exist with respect to $\auxAlph$. 
\end{theorem}

\subsection{Zero-Knowledge}

\begin{definition}[$(\ec, \es, \ezk)$-ZK]
    An $(\ec, \es, \ezk)$-ZK argument for an $\NP$ relation $\Rel$ is a pair of uniform PPT interactive algorithms $(\Pro, \Ver)$ such that
    \begin{itemize}
        \item \textbf{Syntax}
        \begin{itemize}
            \item $\Pro$ is given input $(x, w)$ and $\Ver$ is given input $x$.
            \item $\Pro$ and $\Ver$ interact over the course of $k$ rounds, creating a transcript $\transcript$.
            \item After the interaction, $\Ver$ outputs a single bit.  We let $\langle \Pro(w), \Ver\rangle(x)$ denote this output.
        \end{itemize}

        \item \textbf{Completeness}: For every $(x, w) \in \Rel$,
        \begin{equation*}
            \P\left[\langle \Pro(w), \Ver\rangle(x) = 1\right] \geq 1 - \ec(|x|)
        \end{equation*}

        \item \textbf{Soundness}: For any \emph{non-uniform} PPT algorithm $\widetilde{\Pro}$, we have that for $x \not \in \Lang_\Rel$ such that $|x|$ is large enough,
        \begin{equation*}
            \P\left[\langle \widetilde{\Pro}, \Ver\rangle(x) = 1\right] \leq \es(|x|)
        \end{equation*}

        \item \textbf{Zero-Knowledge}: There exists a PPT simulator $\Sim$ such that for all non-uniform PPT distinguishers $\Distinguisher$ and all $(x, w) \in \Rel$ with $|x|$ large enough,
        \begin{equation*}
            \left| \P\left[\Distinguisher(r, \transcript) = 1 \mid (r, \transcript) \gets \Sim(x)\right] - \P\left[\Distinguisher(r, \transcript) = 1 \middle| \begin{matrix}
                r \samp \zo^m \\
                \transcript \gets \langle \Pro(w), \Ver \rangle (x;r)
            \end{matrix}\right] \right| \leq \ezk(|x|)
        \end{equation*}
        where we let $m$ be the number of random bits used by $\Ver$ and $\langle \Pro(w), \Ver \rangle (x;r)$ denote the transcript generated by $\Pro$ and $\Ver$ interacting, fixing $\Ver$'s random tape to $r$.
    \end{itemize}
\end{definition}

\begin{remark}
    We say a ZK argument for $\Lang$ is \emph{public-coin} if the every message from $\Ver$ is simply a fresh random string.  In particular, note that this implies without loss of generality that $\Ver$ does not maintain an internal state between rounds. 
\end{remark}

\begin{definition}[$(\ec, \es, \ezk)$-NIZK]
    An $(\ec, \es, \ezk)$-NIZK for an $\NP$ relation $\Rel$ is a tuple of polynomial-time algorithms $(\Gen, \Pro, \Ver)$ such that
    \begin{itemize}
        \item \textbf{Syntax}
        \begin{itemize}
            \item $\crs \gets \Gen(1^n)$: $\Gen$ is a randomized algorithm that on input an instance size (in unary) outputs a common reference string.
            \item $\prf \gets \Pro(\crs, x, w)$: $\Pro$ is a randomized algorithm that on input a crs, instance, and witness for that instance outputs a proof.
            \item $b \gets \Ver(\crs, x, \prf)$: $\Ver$ is a \emph{deterministic} algorithm that on input a crs, instance, and proof for that instance outputs a bit denoting acceptance or rejection.
        \end{itemize}
        
        \item \textbf{Completeness}: For every $(x, w) \in \Rel$,
        \begin{equation*}
            \P\left[\Ver(\crs, x, \prf) = 1 \middle| \begin{matrix}
                \crs \gets \Gen(1^{|x|}) \\
                \prf \gets \Pro(\crs, x, w)
            \end{matrix}\right] \geq 1 - \ec(|x|)
        \end{equation*}

        \item \textbf{Soundness}: For any \emph{non-uniform} PPT algorithm $\widetilde{\Pro}$, we have that for $x \not \in \Lang_\Rel$ such that $|x|$ is large enough,
        \begin{equation*}
            \P\left[\Ver(\crs, x, \tilde{\prf}) = 1 \middle| \begin{matrix}
                \crs \gets \Gen(1^{|x|}) \\
                \widetilde{\prf} \gets \widetilde{\Pro}(\crs, x)
            \end{matrix}\right] \leq \es(|x|)
        \end{equation*}

        \item \textbf{Zero-Knowledge}: There exists a PPT simulator $\Sim$ such that for all non-uniform PPT distinguishers $\Distinguisher$ and all $(x, w) \in \Rel$ with $|x|$ large enough,
        \begin{equation*}
            \left| \P\left[\Distinguisher(\crs, \prf) = 1 \mid (\crs, \prf) \gets \Sim(x)\right] - \P\left[\Distinguisher(\crs, \prf) = 1 \middle| \begin{matrix}
                \crs \gets \Gen(1^{|x|}) \\
                \prf \gets \Pro(\crs, x, w)
            \end{matrix}\right] \right| \leq \ezk(|x|)
        \end{equation*}
        %\begin{equation*}
            %\Sim(x) \overset{c}{\approx}_{\ezk(|x|)} \left( \crs \gets \Gen(1^{|x|}), \prf \gets \Pro(\crs, x, w) \right)
        %\end{equation*}
    \end{itemize}
\end{definition}

\begin{remark}
    We will often refer to a (non-interactive) zero-knowledge argument as being for the language defined by the relation $\Rel$, rather than for the relation itself.
\end{remark}

\begin{remark}
    We have defined soundness for NIZKs as being \emph{computational} and \emph{non-adaptive}.  One can also consider \emph{statistical} soundness, where our guarantee needs to hold even against inefficient $\widetilde{\Pro}$.  Additionally, one can consider \emph{adaptive} soundness, where the adversary is allowed to choose the instance after seeing the CRS.  Formally, our soundness requirement would be replaced by
    \begin{equation*}
            \P\left[ \begin{matrix}
                |x| = n \land x \not \in \Lang_\Rel \\
                \land \ \Ver(\crs, x, \tilde{\prf}) = 1
            \end{matrix}
            \middle| \begin{matrix}
                \crs \gets \Gen(1^{n}) \\
                (x, \tilde{\prf}) \gets \tilde{\Pro}(\crs)
            \end{matrix}\right] \leq \es(n)
    \end{equation*}
    for sufficiently large $n$.
\end{remark}

\subsection{Complexity Classes} \label{sec:def-comp-class}
\label{sec:defcc}
\begin{definition}[Infinitely-Often BPP]\label{def:io-bpp}
    Let $\ioBPP$ be the class of languages $\Lang$ such that there exists a PPT algorithm $\Alg$, functions $a, b: \mathbb{N} \to [0, 1]$, and a polynomial $p$ such that for infinitely many $n$,
    \begin{itemize}
        \item For all $x \in \Lang \cap \zo^n$, $\P[\Alg(x) = 1] \geq a(n)$
        \item For all $x \in \zo^n - \Lang$, $\P[\Alg(x) = 1] \leq b(n)$
        \item $a(n) - b(n) \geq \frac{1}{p(n)}$
    \end{itemize}
    We let $\ioBPPpoly$ be defined similarly, except that $\Alg$ is allowed to be non-uniform.
\end{definition}

\begin{definition}[Infinitely-Often $\Ppoly$]
    Let $\ioPpoly$ be the class of languages $\Lang$ such that there exists a non-uniform polynomial-time algorithm $\Alg$ such that for infinitely many $n$,
    \begin{itemize}
        \item For all $x \in \Lang \cap \zo^n$, $\Alg(x) = 1$
        \item For all $x \in \zo^n - \Lang$, $\Alg(x) = 0$
    \end{itemize}

\end{definition}

\begin{remark} \label{rmk:bpppoly-vs-ppoly}
    The textbook proof that $\BPP \subseteq \Ppoly$ also says that $\ioBPPpoly \subseteq \ioPpoly$, so in particular $\ioBPPpoly = \ioPpoly$.  
    Since $\ioPpoly$ is a more standard complexity class, we will state our results in terms of it, even when the algorithms we give are in fact probabilistic.
\end{remark}

\subsection{Chernoff Bounds}

We will need a few different standard instantiations of Chernoff bounds for binomial random variables, as listed below.

\begin{theorem}[Chernoff Bounds] \label{thm:chernoff-combined}
    For every $i \in [m]$, let $X_i$ be an independent Bernoulli random variable that takes value $1$ with probability $p$, and let $X:= \sum_i X_i$.  Then:
    \begin{enumerate}
        \item \label{eqn:chernoff-mult-above} for any $0 < \delta < 1$, $\Pr[X \geq (1 + \delta) mp] \leq \exp\left(-\frac{\delta^2mp}{3}\right)$.
        \item \label{eqn:chernoff-mult-below} for any $0 < \delta < 1$, $\Pr[X \leq (1 - \delta) mp] \leq \exp\left(-\frac{\delta^2mp}{2}\right)$.
        \item \label{eqn:chernoff-add} for any $0 < \Delta$, $\Pr[|X - mp| \geq \Delta] \leq 2\exp\left(-\frac{\Delta^2}{m}\right)$.
    \end{enumerate}
    
\end{theorem}

\iffalse{

\begin{theorem}[Additive Chernoff Bound]\label{thm:chernoff-additive}
    For every $i \in [m]$, let $X_i$ be an independent Bernoulli random variable that takes value $1$ with probability $p$. 
    Let $X:= \sum_i X_i/m$. Then for $\Delta>0$:
    \[
        \Pr[|X - p| \geq \Delta] \leq 2e^{-m\cdot\Delta^2}
    \]
\end{theorem}

\begin{theorem}[Multiplicative Chernoff Bound]\label{thm:chernoff-mult}
    For every $i \in [m]$, let $X_i$ be an independent Bernoulli random variable that takes value $1$ with probability $p$. 
    Let $X:= \sum_i X_i/m$. Then for $\Delta>0$:
    \[
        \Pr[|X - p| \geq \Delta p] \leq 2e^{-\frac{\Delta^2 \cdot mp}{3}}
    \]
\end{theorem}

}\fi

\subsection{Lemmas From \cite{LMP}}

In order to bridge the gap from auxiliary-input one-way functions to standard one-way functions, we recall the following definition and lemmas from~\cite{HirNan} as presented in \cite{LMP}.

\begin{definition}[Infinitely-Often Average-Case BPP, \cite{LMP} Definition 2.1] \label{def:ioAvgBPP}
    Let $\Lang$ be a language and $\Dist = \{ \Dist_n \}_{n \in \mathbb{N}}$ be an efficiently sampleable distribution.  We say $(\Lang, \Dist) \in \aBPPpoly$ if there exists a non-uniform PPT $\Alg$ such that for infinitely many $n \in \mathbb{N}$:
    \begin{itemize}
        \item For all $x \in \Supp(\Dist_n)$, $\Pr[\Alg(x) \in \{ \fail, \Lang(x) \}] \geq 0.9$
        \item $\Pr[\Alg(x) = \fail \mid x \samp \Dist_n] \leq 0.25$
    \end{itemize}
\end{definition}

\begin{lemma}[\cite{LMP} Lemma 3.5] \label{lem:lmp-aiowf-to-avghard}
    Suppose auxiliary-input one-way functions exist.  Then there is a language $\Lang \in \NP$ and a polynomial $m(\cdot)$ such that $(\Lang, \{ U_{m(n)} \}_{n \in \mathbb{N}}) \not \in \aBPPpoly$.
\end{lemma}

In order to fit the next lemma from \cite{LMP} into our framework, it will help to introduce some new notation.

\begin{definition} [Decision-To-Inversion Reduction] \label{def:dti-reduction}
    Fix a language $\Lang$.  A decision-to-inversion reduction for $\Lang$ is a PPT machine $R^{(\cdot)}$, a function family $\{ f_x \}_{x \in \zo^*}$, and a polynomial $p(\cdot)$ such that for any non-uniform PPT $\Alg$ and sufficiently large\footnote{Note here that we explicitly allow the notion of ``sufficiently large'' to depend on $\Alg$.} $x$ such that $\Alg(x, \cdot)$ inverts $f_x$ with probability at least $\frac{1}{p(|x|)}$, $\Pr[R^\Alg(x) = \Lang(x)] \geq 1 - \frac{1}{p(|x|)}$.
\end{definition}

\begin{lemma}[Implicit in proof of \cite{LMP} Lemma 3.1] \label{lem:lmp-reduction-to-owf}
    Let $(\Lang, \Dist) \not \in \aBPPpoly$.  If there exists a decision-to-inversion reduction for $\Lang$ with $p(n) = \omega(1)$, then one-way functions exist.
\end{lemma}

Since Lemma \ref{lem:lmp-reduction-to-owf} is not stated in exactly this form in \cite{LMP}, we include a formal proof of it in Appendix \ref{sec:lmp-proof}, noting that the proof essentially mirrors what already appears in \cite{LMP}.

\section{OWF From Non-Trivial NIZK}

We begin with the simpler case of NIZKs.  

\begin{theorem} \label{thm:nizk-to-owf-wc}
    Suppose there exists a language $\Lang \not \in \ioPpoly$ such that $\Lang$ has an $(\ec, \es, \ezk)$-NIZK with $\ec + \es + \ezk <_n 1$.  Then there exists an auxiliary-input one-way function.
\end{theorem}

\begin{proof}
    Let $(\Gen, \Pro, \Ver)$ be the $(\ec, \es, \ezk)$-NIZK for $\Lang$ guaranteed by the theorem statement, and $p(n)$ be a polynomial such that for sufficiently large $n$, $\ec(n) + \es(n) + \ezk(n) < 1 - \frac{1}{p(n)}$.\footnote{In what follows, it will help to assume WLOG that $p(n) = \omega(n^2)$.}  Suppose for the sake of contradiction that auxiliary-input one-way functions don't exist.  Define:
    \begin{itemize}
        \item $M$ as an auxiliary-input sampler where $M(x;r)$ computes $(\crs, \prf) \gets \Sim(x;r)$ and outputs $\crs$.

        \item $N$ as the io-UE machine guaranteed by Theorem \ref{thm:owf-vs-ioue} with respect to sampler $M$ and polynomial $20p(n)$.

        \item $\paramSet$ as the (infinite) set of $n$ such that the guarantees of $N$ hold.
    \end{itemize}
    We now claim that Algorithm \ref{alg:nizk-decider} decides $\Lang$ infinitely often, contradicting that $\Lang \not \in \ioPpoly$.

    \begin{algorithm}
        \DontPrintSemicolon
        \caption{Algorithm $\Alg$ to decide $\Lang$} \label{alg:nizk-decider}
        \KwIn{$x \in \zo^{n}$}
        \KwOut{Decision bit $b$, with $b = 1$ representing $x \in \Lang$}
        Compute $\crs \gets \Gen(1^n)$ \;
        \For {$\_ \gets 1$ \KwTo $20n\cdot p(n)$}{
            Compute $r' \gets N(x, \crs)$ \;
            Compute $(\crs', \prf') \gets \Sim(x;r')$ \;
            \If{$\Ver(\crs, x, \prf') = 1$}{
                Output $b = 1$
            }
        }
        Output $b = 0$
    \end{algorithm}

    To analyze this algorithm, we first consider what it does when $x \in \zo^n - \Lang$.  In this case, $\Alg$ can only accept if it finds an accepting proof for $x$ relative to an honestly-generated CRS.  Since $\Alg$ is efficient, soundness tells us that this happens with probability at most $\es(n)$.

    Next, we consider what $\Alg$ does when $x \in \zo^n \cap \Lang$ for some $n \in \paramSet$.  Let $p_\crs$ be the probability that a single iteration of the for loop in $\Alg$ accepts.  We call $\crs$ ``bad'' if $p_{\crs} \leq \frac{1}{20p(n)}$, and ``good'' if $p_{\crs} \geq \frac{1}{10p(n)}$.  The core of our analysis comes from the following claim:

    \begin{claim} \label{clm:bad-crs}
        Suppose $x \in \zo^n \cap \Lang$ for some $n \in \paramSet$.  Then $\P[\crs \text{ is bad} \mid \crs \gets \Gen(1^n)] \leq \ezk(n) + \ec(n) + \frac{1}{4p(n)}$.
    \end{claim}

    \begin{proof}
        \begin{algorithm}
            \DontPrintSemicolon
            \caption{Distinguisher $\Distinguisher$} \label{alg:nizk-distinguisher}
            \KwIn{$x \in \zo^{n}$, $(\crs, \prf)$}
            \KwOut{Decision bit $b$, with $b = 0$ representing that $(\crs, \prf)$ was generated by $\Sim$}
            \For {$\_ \gets 1$ \KwTo $20n\cdot p(n)$}{
                Compute $r' \gets N(x, \crs)$ \;
                Compute $(\crs', \prf') \gets \Sim(x;r')$ \;
                Run $\Ver(\crs, x, \prf')$ \;
            }
            Let $\text{count}$ be the number of iterations where $\Ver$ accepted \;
            \uIf{$\text{count} \leq 1.5n$ and $\Ver(x, \crs, \prf) = 1$} {
                Output $b = 1$
            }
            \Else{
                Output $b = 0$
            }
        \end{algorithm}
    
        Consider the distinguisher $\Distinguisher$ defined by Algorithm \ref{alg:nizk-distinguisher}.  We first claim that if $\Distinguisher$ is given $(\crs, \prf)$ drawn from $\Sim(x)$, it will output 1 with probability at most $\frac{1}{5p(n)}$.  First note that if we change $\prf$ to instead be the proof we get from using randomness $r' \gets N(x, \crs)$, our guarantees on $N$ tell us that this probability can change by at most $\frac{1}{20p(n)}$.  Now there are two cases to consider: either $\crs$ is good, or it is not good.

        In the case where $\crs$ is good, we note that each iteration of the for loop has $\Ver(\crs, x, \prf') = 1$ with probability at least $\frac{1}{10p(n)}$.  Thus, bound \ref{eqn:chernoff-mult-below} in Theorem \ref{thm:chernoff-combined} (with respect to $\delta = \frac{1}{4}$) tells us that the final value of $\text{count}$ can be below $1.5n$ with probability at most $\exp(-n/16)$, which is at most $\frac{1}{20p(n)}$ for sufficiently large $n$.

        \iffalse{
        In the case where $\crs$ is good, we note that each iteration of the for loop has $\Ver(\crs, x, \prf') = 1$ with probability at least $\frac{1}{10p(n)}$.  Thus, the expected final value of $\text{count}$ is at least $2n$, so a simple Chernoff bound tells us that the probability it is at most $1.5n$ is bounded by $\exp(-n/16)$.  Since $\Distinguisher$ can only output 1 if $\text{count}$ is at most $1.5n$, we get that this happens when $\crs$ is good with probability at most $\exp(-n/16)$, which is at most $\frac{1}{20p(n)}$ for sufficiently large $n$.
        }\fi

        In the case where $\crs$ is not good, we know that the probability that we get an accepting proof by computing $r' \gets N(x, \crs)$ and $(\crs', \prf') \gets \Sim(x;r')$ is at most $\frac{1}{10p(n)}$.  But note that this is exactly how our given proof is computed!  Since $\Distinguisher$ can only output 1 if $\Ver(x, \crs, \prf) = 1$, we have that when $\crs$ is not good this happens with probability at most $\frac{1}{10p(n)}$.  Combining all of the above, we get that $\Distinguisher$ outputs 1 on a simulated $(\crs, \prf)$ with probability at most $\frac{1}{5p(n)}$ as desired.

        Now we turn to what $\Distinguisher$ does when given $(\crs, \prf)$ generated by $\Gen$ and $\Pro$.  We claim here that it will output 1 with probability at least $\P[\crs \text{ is bad} \mid \crs \gets \Gen(1^n)] - \ec(n) - \frac{1}{20p(n)}$.  Indeed, there are only two ways for $\Distinguisher$ to output 0: either $\Ver(x, \crs, \prf)$ rejects, or the final value of $\text{count}$ is at most $1.5n$.  The former can happen with probability at most $\ec(n)$ by completeness.  For the latter, if $\crs$ is bad, we know that each iteration of the for loop has $\Ver(\crs', x, \prf') = 1$ with probability at most $\frac{1}{20p(n)}$.  Pessimistically assuming that this probability is \emph{exactly} $\frac{1}{20p(n)}$, bound \ref{eqn:chernoff-mult-above} of Theorem \ref{thm:chernoff-combined} (with respect to $\delta = \frac{1}{2}$) tells us that the final value of $\text{count}$ can exceed $1.5n$ with probability at most $\exp(-n/12)$; for sufficiently large $n$, this can be bounded by $\frac{1}{20p(n)}$.

        \iffalse{
        Now we turn to what $\Distinguisher$ does when given $(\crs, \prf)$ generated by $\Gen$ and $\Pro$.  We claim here that it will output 1 with probability at least $\P[\crs \text{ is bad} \mid \crs \gets \Gen(1^n)] - \ec(n) - \frac{1}{20p(n)}$.  Indeed, there are only two ways for $\Distinguisher$ to output 0: either $\Ver(x, \crs, \prf)$ rejects, or the final value of $\text{count}$ is at most $1.5n$.  The former can happen with probability at most $\ec(n)$ by completeness.  For the latter, if $\crs$ is bad, we know that each iteration of the for loop has $\Ver(\crs', x, \prf') = 1$ with probability at most $\frac{1}{20p(n)}$.  Thus, the expected final value of $\text{count}$ is at most $n$, so a simple Chernoff bound tells us that the probability it is at least $1.5n$ is bounded by $\exp(-3n/4)$; for sufficiently large $n$, this can be bounded by $\frac{1}{20p(n)}$.
        }\fi

        Putting this all together, we have that $\Distinguisher$ distinguishes a real $(\crs, \prf)$ from a simulated one with advantage at least $\P[\crs \text{ is bad} \mid \crs \gets \Gen(1^n)] - \ec(n) - \frac{1}{4p(n)}$.  But by zero-knowledge, we know that this advantage can be at most $\ezk(n)$.  Rearranging terms gives the desired claim.
    \end{proof}

    We can now argue that $\Alg$ must accept $x \in \Lang$ with good probability.  Suppose that $\crs$ is not bad.  Then each iteration of the for loop must have $\Ver(\crs, x, \prf') = 1$ with probability at least $\frac{1}{20p(n)}$.  Thus, the probability that \emph{no} iteration of the for loop outputs 1 is at most $\left( 1 - \frac{1}{20p(n)} \right)^{20n \cdot p(n)} \leq \exp(-n)$, which for sufficiently large $n$ is at most $\frac{1}{4p(n)}$.  Thus, the probability that $\Alg$ accepts $x \in \zo^n \cap \Lang$ for $n \in \paramSet$ is at least the probability that $\crs$ is not bad minus $\frac{1}{4p(n)}$; plugging in Claim \ref{clm:bad-crs}, this is at least $1 - \ezk(n) - \ec(n) - \frac{1}{2p(n)}$.

    Finally, this all means that we can instantiate Definition \ref{def:io-bpp} with respect to $a(n) = 1 - \ezk(n) - \ec(n) - \frac{1}{2p(n)}$ and $b(n) = \es(n)$.  The above shows that we satisfy the first two points of the definition for all $n \in \paramSet$.  For sufficiently large $n$, we know that $\ec(n) + \es(n) + \ezk(n) < 1 - \frac{1}{p(n)}$, so $a(n) - b(n) > \frac{1}{2p(n)}$, satisfying the third point.  Thus, we have that $\Lang \in \ioBPPpoly$, so by Remark \ref{rmk:bpppoly-vs-ppoly}, $\Lang \in \ioPpoly$.
\end{proof}

\begin{remark} \label{rmk:nizk-owf-reduction}
    Theorem \ref{thm:nizk-to-owf-wc} can also be phrased in the following way: if $\Lang$ has an $(\ec, \es, \ezk)$-NIZK with $\ec + \es + \ezk <_n 1$, then $\Lang$ has a decision-to-inversion reduction (Definition \ref{def:dti-reduction}).  Indeed, all we need to prove the correctness of Algorithm \ref{alg:nizk-decider} is for $N$ to satisfy the io-UE guarantees on auxiliary input $x$.  Additionally, we note that the proof of Theorem \ref{thm:owf-vs-ioue} constructs a candidate auxiliary-input one-way function $\{ f_x \}_x$ and reduces inverting $f_x$ to satisfying io-UE guarantees on auxiliary input $x$.  Putting these two together gives the desired reduction $R$.\footnote{Technically, Algorithm \ref{alg:nizk-decider} doesn't quite give a decision-to-inversion reduction as written, since we are only guaranteed a polynomial gap between the probability of accepting $x \in \Lang$ versus $x \not \in \Lang$.  However, standard BPP amplification can convert this into a reduction that outputs the correct answer with high probability.}
\end{remark}

\iffalse{
\begin{remark} \label{rmk:nizk-owf-reduction}
    Theorem \ref{thm:nizk-to-owf-wc} can also be phrased in the following way: given an $(\ec, \es, \ezk)$-NIZK for $\Lang$ and $p(\cdot)$ and $q(\cdot)$ as any two polynomials, there exists a polynomial-time oracle-aided reduction $R^{(\cdot)}$ and a family of functions $\{ f_x \}_x$ such that for any efficient (possibly non-uniform) $\Alg$,
    \begin{itemize}
        \item For all $x \not \in \Lang$, $\Pr[R^\Alg(x) = 1] \leq \es(|x|)$
        \item For all $x \in \Lang$, if $\Alg$ inverts $f_x$ with probability at least $\frac{1}{p(|x|)}$, then $\Pr[R^\Alg(x) = 1] \geq 1 - \ec(|x|) - \ezk(|x|) - \frac{1}{q(|x|)}$.
    \end{itemize}
    Indeed, all we need in order for Algorithm \ref{alg:nizk-decider} to satisfy the given conditions if for $N$ to satisfy the io-UE guarantees on auxiliary input $x$.  Additionally, we note that the proof of Theorem \ref{thm:owf-vs-ioue} constructs a candidate auxiliary-input one-way function $\{ f_x \}_x$ and reduces inverting $f_x$ to satisfying io-UE guarantees on auxiliary input $x$.  Putting these two together gives the desired reduction $R$.\footnote{Technically, certain parts of our proof only work as long as $|x|$ is sufficiently large.  However, this is simple to fix by just having $R$ decide $x$ using brute force whenever $|x|$ is too small.}
\end{remark}
}\fi

Given the above remark, the techniques from \cite{HirNan,LMP} immediately allow us to improve Theorem \ref{thm:nizk-to-owf-wc} from an \emph{auxiliary-input} one-way function to a standard one-way function.

\begin{corollary} \label{cor:nizk-to-owf}
    Suppose that every $\Lang \in \NP$ has an $(\ec, \es, \ezk)$-NIZK with $\ec + \es + \ezk <_n 1$ and $\NP \not \subseteq \ioPpoly$.  Then OWFs exist.
\end{corollary}

\begin{proof}
    Since $\NP \not \subseteq \ioPpoly$, we know there is a language $\Lang \in \NP$ (that thus has an appropriate NIZK) where $\Lang \not \in \ioPpoly$.  Theorem \ref{thm:nizk-to-owf-wc} thus tells us that there exists an auxiliary-input one-way function.  We can then plug this into Lemma \ref{lem:lmp-aiowf-to-avghard} to get a new language $\Lang' \in \NP$ and a polynomial $m(\cdot)$ such that $(\Lang', \{ U_{m(n)} \}_{n \in \mathbb{N}}) \not \in \aBPPpoly$.  But since $\Lang' \in \NP$ (and thus has an appropriate NIZK), Remark \ref{rmk:nizk-owf-reduction} tells us that there is a decision-to-inversion reduction for $\Lang'$.  Hence, Lemma \ref{lem:lmp-reduction-to-owf} tells us that one-way functions exist.
\end{proof}

\section{OWF From Non-Trivial Public-Coin ZK}

We now turn to the more general case of interactive zero-knowledge.

\begin{theorem} \label{thm:zk-to-owf-wc}
    Suppose there exists a language $\Lang \not \in \ioPpoly$ such that $\Lang$ has a constant-round, public coin $(\ec, \es, \ezk)$-ZK argument with $\ec + \es + \ezk <_n 1$.  Then there exists an auxiliary-input one-way function.
\end{theorem}

Before proving this theorem, it will help to define some tools that will assist us in the proof.  In all of the sections below, let
\begin{itemize}
    \item $(\Gen, \Pro, \Ver)$ be the $(\ec, \es, \ezk)$-ZK argument for $\Lang$ guaranteed by the theorem statement with simulator $\Sim$,
    \item $k$ be the (constant) number of rounds of interaction in this argument,
    \item $p(n)$ be a polynomial such that for sufficiently large $n$, $\ec(n) + \es(n) + \ezk(n) < 1 - \frac{1}{p(n)}$,
    \item $p_\Est(n)$ as the polynomial $256k^2 \cdot n \cdot p(n)^2$,
    \item $M$ as an auxiliary-input sampler where $M(x;r)$ interprets $r = r_1 \| r_2$ for $r_1 \in [k]$ and outputs the first $r_1 - 1$ rounds of the transcript generated by $\Sim(x;r_2)$,
    \item $N$ be the io-UE machine guaranteed (assuming auxiliary-input one-way functions don't exist) by Theorem \ref{thm:owf-vs-ioue} with respect to sampler $M$ and polynomial $8k^2 \cdot p(n)$, and
    \item $\paramSet$ as the (infinite) set of $n$ such that the guarantees of $N$ hold.
\end{itemize}

\subsection{Prover $\widetilde{\Pro}$}
We start by defining a malicious prover strategy $\widetilde{\Pro}$ for $(\Gen, \Pro, \Ver)$.   Intuitively, $\widetilde{\Pro}$ will, given the current partial transcript of interaction with $\Ver$, use $N$ to sample ``many'' possible next prover messages.  In order to choose which to use, $\widetilde{\Pro}$ will estimate the probability of each leading to an accepting transcript if it continues to use this strategy in future rounds, and pick the one with the largest success probability.  We formally define the behavior of $\widetilde{\Pro}$ in Algorithm \ref{alg:zk-p-tilde}; we've separated out the definition of the success probability estimation procedure $\Est$ as Algorithm \ref{alg:zk-prob-est}, as this will see use elsewhere.

\begin{algorithm}
    \DontPrintSemicolon
    \caption{Malicious Prover $\widetilde{\Pro}$} \label{alg:zk-p-tilde}
    \KwIn{$x \in \zo^{n}$, partial transcript $\transcript$ with $\Pro$ sending the next message}
    \KwOut{$m$ as the next message to send}
    Let $i$ be the number of rounds already completed in $\transcript$ \;
    \For {$\_ \gets 1$ \KwTo $16k \cdot n \cdot p(n)$}{
        Compute $r' \gets N(x, \transcript)$ \;
        Parse $r' = r_1' \| r_2'$ \;
        Compute $\transcript'$ as the first $i + 1$ messages of $\Sim(x;r_2')$ \;
        Compute $\Est(x, \transcript')$ \;
    }
    Let $\transcript^*$ be the transcript with the highest value of $\Est(x, \transcript')$ sampled above \;
    Output $m^*$ as the $(i + 1)$st message in $\transcript^*$
\end{algorithm}

\begin{algorithm}
    \DontPrintSemicolon
    \caption{Success Probability Estimator $\Est$} \label{alg:zk-prob-est}
    \KwIn{$x \in \zo^{n}$, (possibly partial) transcript $\transcript$}
    \KwOut{Estimate of the probability that $\Ver$ will accept when interacting with $\widetilde{\Pro}$ starting on transcript $\tau$}
    \For {$\_ \gets 1$ \KwTo $12p_\Est(n)^4$}{
        Simulate\footnotemark \ interaction between $\Ver$ and $\widetilde{\Pro}$ starting from partial transcript $\transcript$ \;
    }
    Output fraction of iterations where $\Ver$ accepted \;
\end{algorithm}

\footnotetext{We note that this simulation is possible because $\Ver$ is public-coin, and thus has no internal state.  In fact, this is the \emph{only} reason we need to limit our proof to public-coin protocols.}

\begin{remark}
    We note that while the definitions of $\widetilde{\Pro}$ and $\Est$ refer to each other, $\widetilde{\Pro}$ always calls $\Est$ on a strictly longer transcript than it started with, so these definitions are not circular. 
\end{remark}

We now prove some useful properties of $\widetilde{\Pro}$ and $\Est$.

\begin{lemma} \label{lem:pro-est-efficient}
   $\widetilde{\Pro}$ and $\Est$ run in polynomial time.
\end{lemma}

\begin{proof}
    Define $T_{\widetilde{\Pro}}(i)$ as the run time of $\widetilde{\Pro}$ when it is given a transcript with $i$ messages in it.  We note that this is a decreasing function, as $\widetilde{\Pro}$
    needs to do less work estimating its success probability when the transcript is closer to completion.
    
    Our first claim is that there exists a polynomial $q(\cdot, \cdot)$ such that $T_{\widetilde{\Pro}}(i) \leq q(k, n) \cdot T_{\widetilde{\Pro}}(i + 2)$ and $T_{\widetilde{\Pro}}(k - 1) \leq q(k, n)$.  Indeed, ignoring various (additive) polynomial overheads, when $\widetilde{\Pro}$ is given a transcript with $i$ messages, it simply needs to run $\Est$ on transcripts with $i + 1$ messages a total of $16k \cdot n \cdot p(n)$ times.  Each call to $\Est$ needs to simulate $12p_\Est(n)^4$ interactions between $\widetilde{\Pro}$ and $\Ver$, each of which needs to run $\widetilde{\Pro}$ at most $k$ times, each time starting with a transcript that has at least $i + 2$ messages.  Combining these two statements, we get the first part of our claim.
    %\footnote{Technically the additive overheads will mean that we only get that $T_{\widetilde{\Pro}}(i) \leq k^{10} \cdot q_1(n) \cdot T_{\widetilde{\Pro}}(i + 2) + q_2(n)$ for some polynomials $q_1$ and $q_2$.  But since $T_{\widetilde{\Pro}}(i + 2) \geq 1$, we can rewrite this as $T_{\widetilde{\Pro}}(i) \leq k^{10} \cdot (q_1(n) + q_2(n)) \cdot T_{\widetilde{\Pro}}(i + 2)$.}
    For the second part, we note that $\Est$ on a full transcript simply needs to check if $\Ver$ accepts $12p_\Est(n)^4$ times, each of which takes polynomial time independent of $\widetilde{\Pro}$.
    Since $\widetilde{\Pro}$ on input a transcript with $k - 1$ messages simply needs to run $\Est$ on a full transcript $16k \cdot n \cdot p(n)$ times, we get the second half of the claim.

    Given this claim, we can say that $T_{\widetilde{\Pro}}(0) \leq q(k, n)^{O(k)}$.  Since $k$ is a constant\footnote{We note that this is the only reason we need to limit our proof to constant round protocols.} and $T_{\widetilde{\Pro}}$ is a decreasing function, we have that $\widetilde{\Pro}$ always runs in polynomial time.  Furthermore noting that $\Est$ only ever needs to run $\widetilde{\Pro}$ a polynomial number of times (along with some other polynomial additive overhead), we have that $\Est$ also always runs in polynomial time.
\end{proof}

\begin{lemma} \label{lem:zk-est-error}
    On any input $\transcript$, the output of $\Est$ is within an additive distance of $\frac{1}{p_\Est(n)}$ of the true acceptance probability with probability at least $1 - \frac{1}{p_\Est(n)}$.
\end{lemma}
\begin{proof}
    Let $p_\transcript$ 

    be the true acceptance probability given that $\widetilde{\Pro}$ and $\Ver$ start with the transcript $\transcript$. An application of bound \ref{eqn:chernoff-add} in Theorem \ref{thm:chernoff-combined} (with respect to $\Delta = 12p_\Est(n)^3$) tells us that
    \begin{align*}
        \P[|\Est(x, \transcript) - p_\transcript| \geq \frac{1}{p_\Est(n)}] &\leq 2\exp\left( -\frac{144p_\Est(n)^6}{12 p_\Est(n)^4 } \right) \\
        & = 2\exp(-12 p_\Est(n)^2) \leq \frac{1}{p_\Est(n)}
    \end{align*}
\end{proof}

\noindent Before stating our next property, we need to introduce some notation.

\begin{definition} \label{def:p-msg-quality}
    Fix a partial transcript $\transcript$ where the prover is going next, and let $i$ be the number of messages in $\transcript$.  For any possible next prover message $m$, let $p_m$ be the probability that $\Ver$ would accept if it interacted with $\widetilde{\Pro}$ starting on the transcript $\transcript \| m$. We define $q_\text{threshold}$ as the threshold such that $q_\text{threshold} \geq p_{m'}$ for at least $1-\frac{1}{16k \cdot p(n)}$ fraction (by probability) of messages $m'$ sampled using $N$. More formally, let $q_\text{threshold}$ be the smallest value such that
    \begin{align*}
    \P\left[ q_\text{threshold} \geq p_{m'}\ \middle | \begin{array}{c}
         r'= r_1' \| r_2' \gets N(x, \transcript) \\
         \tau' \gets \Sim(x;r_2') \\
         m' = \text{the } (i + 1) \text{st message in } \transcript'
    \end{array} \right] > 1 - \frac{1}{16k \cdot p(n)}
\end{align*}
%     Then we let the \emph{quality} of $m$ be the smallest value $q$ such that
%     \begin{align*}
%     \P\left[ p_m \geq p_{m'} + q \ \middle | \begin{array}{c}
%          r'= r_1' \| r_2' \gets N(x, \transcript) \\
%          \tau' \gets \Sim(x;r_2') \\
%          m' = \text{the } (i + 1) \text{st message in } \transcript'
%     \end{array} \right] > 1 - \frac{1}{16k \cdot p(n)}
% \end{align*}
% (Informally, the quality of a message captures how much it increases the success probability compared to ``most'' messages that could be sampled using $N$; note that $q$ is a real number between $-1$ and $1$.) 
We call $m$ ``okay'' if $p_m \geq q_\text{threshold}-\frac{2}{p_\Est(n)}$, ``good'' if $p_m \geq q_\text{threshold}$, and ``very good'' if $p_m \geq q_\text{threshold}+\frac{2}{p_\Est(n)}$. 
\end{definition}

\iffalse{
For a given partial transcript $\transcript$ with $i$ messages where the prover is going next and a prover message $m$, let $p_m$ be the probability that $\Ver$ would accept if it interacted with $\widetilde{\Pro}$ starting on the transcript $\transcript \| m$.  Call $m$ ``good'' if
\begin{align*}
    \P\left[ p_m \geq p_{m'} - \frac{2}{p_\Est(n)} \middle | \begin{array}{c}
         r'= r_1' \| r_2' \gets N(x, \transcript) \\
         \tau' \gets \Sim(x;r_2') \\
         m' = \text{the } (i + 1) \text{st message in } \transcript'
    \end{array} \right] > 1 - \frac{1}{16k \cdot p(n)}
\end{align*}
(Informally, $m$ is ``good'' if its success probability almost as good as most messages next messages that we would get by using $N$.)  In later sections, we will also refer to $m$ as ``very good'' if
\begin{align*}
    \P\left[ p_m > p_{m'} + \frac{2}{p_\Est(n)} \middle | \begin{array}{c}
         r'= r_1' \| r_2' \gets N(x, \transcript) \\
         \tau' \gets \Sim(x;r_2') \\
         m' = \text{the } (i + 1) \text{st message in } \transcript'
    \end{array} \right] > 1 - \frac{1}{16k \cdot p(n)}
\end{align*}
(Informally, $m$ is ``very good'' if its success probability enough \emph{better} than most messages next messages that we would get by using $N$.  Note that any ``very good'' message is also in particular ``good''.)
}\fi

\begin{lemma} \label{lem:zk-p-tilde-good}
    Let $\transcript$ be any partial transcript where the prover sends the next message.  If we run $\widetilde{\Pro}(x, \transcript)$ to get a next prover message $m$, then $m$ is ``okay'' with probability at least $1 - \frac{1}{16k \cdot p(n)} - e^{-n}$.
\end{lemma}

\begin{proof}
    We note that $\widetilde{\Pro}$ makes $16k \cdot n \cdot p(n)$ calls to $\Est$.  By Lemma \ref{lem:zk-est-error} and a union bound, we know that the probability that any of these calls is more than $\frac{1}{p_\Est(n)}$ away from the true success probability is at most $\frac{16k \cdot n \cdot p(n)}{p_\Est(n)} = \frac{1}{16k \cdot p(n)}$.

    \iffalse{
    Suppose now we order all possible messages $m$ a single iteration of the for loop in $\widetilde{\Pro}$ could sample by $p_m$ (as defined above, breaking ties arbitrarily).  Then we can immediately say that each iteration of the for loop samples a message in the top $\frac{1}{16k \cdot p(n)}$ fraction of these messages with probability $\frac{1}{16k \cdot p(n)}$.  Thus, the probability that \emph{no} iteration does this is $(1 - \frac{1}{16k \cdot p(n)})^{\frac{1}{16k \cdot n \cdot p(n)}} \leq e^{-n}$.
    }\fi

    But we now note that since each iteration of the for loop in $\widetilde{\Pro}$ samples from exactly the distribution of messages used in Definition \ref{def:p-msg-quality}, each iteration chooses a ``good'' message with probability at least $\frac{1}{16k \cdot p(n)}$.  Thus, the probability that \emph{no} iteration chooses a ``good'' message is at most $(1 - \frac{1}{16k \cdot p(n)})^{16k \cdot n \cdot p(n)} \leq e^{-n}$.
    
    Now note that if neither of the above bad events happen, the message $m^*$ output by $\widetilde{\Pro}$ must be ``okay''.\footnote{In the worst case, we could underestimate the success probability of every ``good'' message we sampled by $\frac{1}{p_\Est(n)}$, and choose a different message due to overestimating its success probability by $\frac{1}{p_\Est(n)}$.}  Thus, we have that $\widetilde{\Pro}$ outputs an ``okay'' message with probability at least $1 - \frac{1}{16k \cdot p(n)} - e^{-n}$, as desired.
\end{proof}

\subsection{Distinguisher $\Distinguisher$}

We next define a distinguisher $\Distinguisher$ that we will use in our proof.  Intuitively, $\Distinguisher$ will check each prover message to see if it is ``much better'' than ``most'' next messages that we would get by using $N$.  In our proof, we will use the fact that this distinguisher can only have $\ezk(n)$ advantage in distinguishing a real from a simulated transcript in order to argue that our algorithm cannot lose too much success probability.  We formally define $\Distinguisher$ as Algorithm \ref{alg:zk-dist}.

\begin{algorithm}
    \DontPrintSemicolon
    \caption{Distinguisher $\Distinguisher$} \label{alg:zk-dist}
    \KwIn{$x \in \zo^{n}$, (full) transcript $\transcript$}
    \KwOut{Decision bit $b$, with $b = 0$ representing that $\transcript$ was generated by $\Sim$}
    \For {each $\Pro$ round $i$}{
        Let $\transcript_i$ be the first $i$ rounds of $\transcript$ and $\transcript_{i - 1}$ the first $i - 1$ rounds \;
        Compute $\Est(x, \transcript_i)$ \;
        \For {$\_ \gets 1$ \KwTo $8k \cdot n \cdot p(n)$}{
            Compute $r' \gets N(x, \transcript_{i - 1})$ \;
            Parse $r' = r_1' \| r_2'$ \;
            Compute $\transcript'$ as $\transcript_{i-1}$ concatenated with the $i$th message of $\Sim(x;r_2')$ \;
            Compute $\Est(x, \transcript')$ \;
        }
        Let $\text{count}$ be the number of iterations where $\Est(x, \transcript') \geq \Est(x, \transcript_i)$ \;
        \If{$\text{count} < n$}{
            Output $b = 1$ \;
        }
    }
    Output $b = 0$ \;
\end{algorithm}

We now prove some useful properties of $\Distinguisher$.

\begin{lemma} \label{lem:zk-dist-sim-prob}
    For large enough $n \in \paramSet$ and $x \in \zo^n$, $\P[\Dist(x, \transcript) = 1 \mid \transcript \gets \Sim(x)] \leq \frac{1}{4 p(n)}$.
\end{lemma}

\begin{proof}
    We first define a sequence of hybrids for how we generate $\transcript$.  For $i \in [k]_0$, let $\Hyb_i(x)$ be defined as in Algorithm \ref{alg:zk-transcript-hyb}.  We note immediately that $\Hyb_k$ is identical to $\Sim$.  Furthermore, the guarantees on $N$ tell us that $\Hyb_i(x)$ has statistical distance at most $\frac{1}{8k \cdot p(n)}$ from $\Hyb_{i + 1}(x)$.\footnote{Note that this loses a factor of $k$ compared to the overall guarantee on $N$ since we are specifically asking $N$ to work on a partial transcript with $i$ messages, where $M$ outputs such a transcript with probability $\frac{1}{k}$.}  Thus, we have that $\Sim(x)$ has statistical distance at most $\frac{1}{8p(n)}$ from $\Hyb_0(x)$.

    \begin{algorithm}
    \DontPrintSemicolon
    \caption{Hybrid Transcript Generation $\Hyb_i$} \label{alg:zk-transcript-hyb}
    \KwIn{$x \in \zo^{n}$}
    \KwOut{(Full) transcript $\transcript$}
    Let $\transcript$ be the first $i$ rounds of $\Sim(x)$ \;
    \For{$j \gets i + 1$ \KwTo $k$}{
        Compute $r' \gets N(x, \transcript)$ \;
        Parse $r' = r_1' \| r_2'$ \;
        Append the $j$th message in $\Sim(x;r_2')$ to $\transcript$ \;
    }
    Output $\transcript$ \;
\end{algorithm}

    To complete the proof, we now claim that $\P[\Dist(x, \transcript) = 1 \mid \transcript \gets \Hyb_0(x)] \leq \frac{1}{8p(n)}$.  Indeed, consider the iteration of $\Distinguisher$ corresponding to prover round $i$.  We note that because $\transcript$ was sampled from $\Hyb_0$, $\transcript_i$ is drawn from the same distribution as each $\transcript'$ that $\Distinguisher$ samples in this iteration.  Thus, if we order $\transcript_i$ and each sampled $\transcript'$ by their respective values of $\Est$ (breaking ties randomly), they will be in a random order.  Since it is only possible for $\Distinguisher$ to output 1 in this iteration if $\transcript_i$ is in the top $n$ of this order, we have that the probability of this happening is at most $\frac{n}{8k \cdot n \cdot p(n)} = \frac{1}{8k \cdot p(n)}$.  Taking a union bound over all prover rounds, we get our desired claim.
\end{proof}

\begin{lemma} \label{lem:zk-dist-catch-good}
    Let $\transcript$ be a (full) transcript and $n = |x|$ sufficiently large.  If any prover message in $\transcript$ is ``very good'' (as in Definition \ref{def:p-msg-quality}, with respect to the partial transcript up to that message), $\Distinguisher(x, \transcript)$ will output $b = 1$ in the iteration corresponding to that round with probability at least $1 - \frac{1}{16k \cdot p(n)}$.
\end{lemma}

\begin{proof}
    Let $i$ be a prover round where the prover's message $m_i$ is ``very good''.  We first note that the corresponding iteration of $\Distinguisher$ only makes $8k \cdot n \cdot p(n) + 1$ calls to $\Est$.  Thus, by Lemma \ref{lem:zk-est-error} and a union bound, we have that all values returned by $\Est$ are within a $\frac{1}{p_\Est(n)}$ additive distance of their corresponding correct values with probability at least $1 - \frac{8k \cdot n \cdot p(n)}{p_\Est(n)} = 1 - \frac{1}{32k \cdot p(n)} - \frac{1}{p_\Est(n)}$.  For the remainder of the proof, we will assume this is the case.

    Since every call to $\Est$ is within an additive $\frac{1}{p_\Est(n)}$ distance of the true acceptance probability, we have that for any $\transcript'$ sampled in the iteration of $\Distinguisher$ corresponding to round $i$, we can only have $\Est(x, \transcript') \geq \Est(x, \transcript_i)$ if $p_{m'} \geq p_{m_i} - \frac{2}{p_\Est(n)}$, where $m'$ is the $i$th message of $\transcript'$.  But since $m_i$ is ``very good'', we know that this happens with probability at most $\frac{1}{16k \cdot p(n)}$ over the sampling of $\transcript'$.  Pessimistically assuming that this probability is exactly $\frac{1}{16k \cdot p(n)}$, bound \ref{eqn:chernoff-mult-above} in Theorem \ref{thm:chernoff-combined} (with respect to $\delta = \frac{1}{2}$) will tell us that the final value of $\text{count}$ can only exceed $\frac{3n}{4}$ with probability at most $\exp(-n/24)$; for sufficiently large $n$, this is at most $\frac{1}{p_\Est(n)}$.  Since $\Distinguisher$ will output 1 if $\text{count} < n$ (and so in particular if $\text{count} \leq \frac{3n}{4}$), this tells us that as long as no call to $\Est$ fails, $\Distinguisher$ will output 1 with probability at least $1 - \frac{1}{p_\Est(n)}$.

    Applying a union bound to the above two paragraphs, we get that overall $\Distinguisher$ outputs 1 with probability at least $1 - \frac{1}{32k \cdot p(n)} - \frac{2}{p_\Est(n)} \geq 1 - \frac{1}{16k \cdot p(n)}$ as desired.
\end{proof}

\subsection{Algorithm $\Alg$}

We can now describe our algorithm for deciding $\Lang$, as described in Algorithm \ref{alg:zk-alg}.

\begin{algorithm}
    \DontPrintSemicolon
    \caption{Algorithm $\Alg$ for $\Lang$} \label{alg:zk-alg}
    \KwIn{$x \in \zo^{n}$}
    \KwOut{Decision bit $b$}
    Simulate an interaction between $\widetilde{\Pro}$ and $\Ver$ on input $x$ \;
    \uIf{$\Ver$ accepts}{
        Output $b = 1$ \;
    }
    \Else{
        Output $b = 0$ \;
    }
\end{algorithm}

\begin{lemma}
    For all $x \in \zo^n - \Lang$, $\P[\Alg(x) = 1] \leq \es(n)$. 
\end{lemma}

\begin{proof}
    By Lemma \ref{lem:pro-est-efficient}, we know that $\widetilde{\Pro}$ is an efficient prover strategy.  Thus, soundness tells us that the probability that $\Ver$ accepts when interacting with $\widetilde{\Pro}$ on $x \not \in \Lang$ is at most $\es(|x|)$.
\end{proof}

\begin{lemma} \label{lem:zk-alg-correctness}
    For all sufficiently large $n \in \paramSet$ and all $x \in \zo^n \cap \Lang$, $\P[\Alg(x) = 1] \geq 1 - \ec(n) - \ezk(n) - \frac{1}{2p(n)}$.
\end{lemma}

\begin{proof}
    Fix $n$ to be large enough that the guarantees of all the above lemmas hold, and $x$ as any element in $\zo^n \cap \Lang$.  We define a sequence of hybrid experiments as follows.  For each $i \in [k]_0$, let $\Hyb_i^1(x)$:
    \begin{itemize}
        \item Create a transcript $\transcript$ by running the honest $\Pro$ and $\Ver$ on $x$ for the first $i$ rounds, and completing the remaining rounds by running $\widetilde{\Pro}$ and $\Ver$.\footnote{As in our definition of $\widetilde{\Pro}$, this is possible because $\Ver$ is stateless.}
        \item Run $\Distinguisher(x, \transcript)$.  For every $j \in [k]_0$, let $R_j$ be the event that $\Distinguisher$ outputs 1 in or before the iteration corresponding to round $j$.
    \end{itemize}
    Furthermore, or each $i \in [k]$, let $\Hyb_i^2(x)$:
    \begin{itemize}
        \item Create a transcript $\transcript$ by running the honest $\Pro$ and $\Ver$ on $x$ for the first $i$ rounds, and completing the remaining rounds by running $\widetilde{\Pro}$ and $\Ver$.
        \item Run $\Distinguisher(x, \transcript)$.  For every $j \in [k]_0$, let $R_j$ be the event that $\Distinguisher$ outputs 1 in or before the iteration corresponding to round $j$.
        \item Create a second transcript $\transcript'$ by taking the first $i - 1$ rounds of $\transcript$ and completing the remaining rounds by running $\widetilde{\Pro}$ and $\Ver$.
    \end{itemize}

    Note that $\Hyb_0^1$ (ignoring $\Distinguisher$) is exactly the same as $\Alg$, so our goal is to prove that $\Ver$ is unlikely to reject in that hybrid.  On the other hand, $\Hyb_k^1$ involves generating an honest transcript, which we know can only be rejected with probability at most $\ec(n)$.  Thus, it suffices for us to show that as we move from $\Hyb_{k}^1$  to $\Hyb_{0}^1$\footnote{The hybrids count down for ease of notation.} the probability that $\Ver$ rejects cannot increase by much. Specifically, we aim to prove
     \begin{claim} \label{clm:zk-hyb-main}
        Fix any $i \in [k]$.  Then we have that
        \begin{equation*}
            \P_{\Hyb_{i - 1}^1}[\Ver(x,\transcript) = 0 \land \overline{R_{i - 1}}] \leq  \P_{\Hyb_i^1}[\Ver(x, \transcript) = 0 \land \overline{R_i}] + \frac{1}{4k \cdot p(n)} + \P_{\Hyb_i^2}[\overline{R_{i - 1}} \land R_i]
        \end{equation*}
    \end{claim}
    \noindent Indeed, by repeated application of the claim, we get that
    \begin{align} \label{eqn:zk-inductive-result}
        \nonumber \P[\Alg(x) = 0] &= \P_{\Hyb_0^1}[\Ver(x, \transcript) = 0] = \P_{\Hyb_0^1}[\Ver(x, \transcript) = 0 \land \overline{R_0}] \\
        \nonumber &\leq \P_{\Hyb_k^1}[\Ver(x, \transcript) = 0 \land \overline{R_k}] + \frac{1}{4p(n)} + \sum_{i \in [k]} \P_{\Hyb_i^2}[ \overline{R_{i - 1}} \land R_i]\\
        &\leq \P_{\Hyb_k^1}[\Ver(x, \transcript) = 0] + \frac{1}{4p(n)} + \sum_{i \in [k]} \P_{\Hyb_i^2}[ \overline{R_{i - 1}} \land R_i]
    \end{align}

    (For the second equality, we used the fact that $R_0$ is always true.)  The first term is at most $\ec(n)$ as noted above.  For the last term, we note that the events $R_i$ and $R_{i - 1}$ are defined over the first $i$ messages of $\transcript$, which are sampled using the honest $\Pro$ and $\Ver$.  Thus, we have that $\P_{\Hyb_i^2}[\overline{R_{i - 1}} \land R_i] = \P_{\text{real}}[\overline{R_{i - 1}} \land R_i]$, where $\text{real}$ is the probability space where we sample the entire transcript $\transcript$ from $\Pro$ and $\Ver$, then run $\Distinguisher(x, \transcript)$.  Since $R_{i - 1} \subseteq R_i$ for all $i \in [k]$, the resulting sum telescopes; that is, $\sum_{i \in [k]} \P_{\text{real}}[\overline{R_{i - 1}} \land R_i] = \P_{\text{real}}[\overline{R_0} \land R_k]$, which by the definition of $R_0$ and $R_k$ is just $\P_\text{real}[\Dist(x, \transcript) = 1]$.  Plugging this all back into (\ref{eqn:zk-inductive-result}), we get
    \begin{align*}
        \P[\Alg(x) = 0] &\leq \ec(n) + \P_\text{real}[\Distinguisher(x, \transcript) = 1] + \frac{1}{4p(n)} \\
        &\leq \ec(n) + \ezk(n) + \P[\Distinguisher(x, \transcript) = 1 \mid \transcript \gets \Sim(x)] + \frac{1}{4p(n)} \\
        &\leq \ec(n) + \ezk(n) + \frac{1}{2p(n)}
    \end{align*}
      where the first inequality comes from zero-knowledge and the second comes from Lemma \ref{lem:zk-dist-sim-prob}. All that now remains is to prove the claim. We do so by proving the following two claims.
    
    \begin{claim} \label{clm:zk-hyb-1}
        Fix any $i \in [k]$.  Then we have that
        \begin{equation*}
            \P_{\Hyb_{i - 1}^1}[\Ver(x,\transcript) = 0 \land \overline{R_{i - 1}}] \leq \P_{\Hyb_i^2}[\Ver(x, \transcript') = 0 \land \overline{R_i}] + \P_{\Hyb_i^2}[\overline{R_{i - 1}} \land R_i]
        \end{equation*}
    \end{claim}

    \begin{claim} \label{clm:zk-hyb-2}
        Fix any $i \in [k]$.  Then we have that
        \begin{equation*}
            \P_{\Hyb_i^2}[\Ver(x,\transcript') = 0 \land \overline{R_i}] \leq \P_{\Hyb_i^1}[\Ver(x, \transcript) = 0 \land \overline{R_i}] + \frac{1}{4k \cdot p(n)}
        \end{equation*}
    \end{claim}
    
    \noindent Note that Claim $\ref{clm:zk-hyb-main}$ trivially follows from Claim $\ref{clm:zk-hyb-1}$ and Claim $\ref{clm:zk-hyb-2}$.

    \begin{proof}[Proof of Claim \ref{clm:zk-hyb-1}]
        We first note that 
        $$\P_{\Hyb_{i-1}^1}[\Ver(x, \transcript) = 0 \land \overline{R_{i - 1}}] = \P_{\Hyb_i^2}[\Ver(x, \transcript') = 0 \land \overline{R_{i - 1}}]$$
        Indeed, in both cases we are checking if $\Ver$ accepts a transcript sampled by taking the first $i - 1$ messages from $\Pro$ and $\Ver$ with the remainder from $\widetilde{\Pro}$ and $\Ver$, as well as if $\Distinguisher$ outputs 1 within the first $i - 1$ messages of that transcript.  But now we can write
        \begin{align*}
            \P_{\Hyb_i^2}[\Ver(x, \transcript') = 0 \land \overline{R_{i - 1}}] &= \P_{\Hyb_i^2}[\Ver(x, \transcript') = 0 \land \overline{R_i}] + \P_{\Hyb_i^2}[\Ver(x, \transcript') = 0 \land \overline{R_{i - 1}} \land R_i] \\
            &\leq \P_{\Hyb_i^2}[\Ver(x, \transcript') = 0 \land \overline{R_i}] + \P_{\Hyb_i^2}[\overline{R_{i - 1}} \land R_i]
        \end{align*}
        which yields the desired claim.
    \end{proof}

    \begin{proof}[Proof of Claim \ref{clm:zk-hyb-2}]
        First note that if $i$ is a verifier round, we have that 
        $$\P_{\Hyb_i^2}[\Ver(x,\transcript') = 0 \land \overline{R_i}] = \P_{\Hyb_i^1}[\Ver(x, \transcript) = 0 \land \overline{R_i}]$$
        Indeed, the only place where these two hybrids could differ is at round $i$.  But note that $\overline{R_i}$ is not affected by this difference, since $\Distinguisher$ only looks at prover rounds.  Additionally, $\transcript$ and $\transcript'$ are identically distributed (both use $\Ver$ to sample the $i$th message), so the event that $\Ver$ rejects is also unaffected.  This means that the statement trivially holds when $i$ is a verifier round, so for the remainder of the proof we will assume that $i$ is a prover round.

        Recalling the definitions of ``okay'' and ``very good'' from Definition \ref{def:p-msg-quality}, we have that
        \begin{align*}
            \P_{\Hyb_i^2}[\Ver(x,\transcript') = 0 \land \overline{R_i}] &= \quad \P_{\Hyb_i^2}[\Ver(x,\transcript') = 0 \land \overline{R_i} \land m_i \text{ is ``very good''}] \\
            &\quad + \P_{\Hyb_i^2}[\Ver(x,\transcript') = 0 \land \overline{R_i} \land m_i \text{ is not ``very good''} \land m_i' \text{ is ``okay''}] \\
            &\quad + \P_{\Hyb_i^2}[\Ver(x,\transcript') = 0 \land \overline{R_i} \land m_i \text{ is not ``very good''} \land m_i' \text{ is not ``okay''}] \\
            &\leq \quad \P_{\Hyb_i^2}[\overline{R_i} \land m_i \text{ is ``very good''}] \\
            &\quad + \P_{\Hyb_i^2}[\Ver(x,\transcript') = 0 \land \overline{R_i} \land m_i \text{ is not ``very good''} \land m_i' \text{ is ``okay''}] \\
            &\quad + \P_{\Hyb_i^2}[m_i' \text{ is not ``okay''}]
        \end{align*}
        where $m_i$ is the $i$th message in $\transcript$ and $m_i'$ is the $i$th message in $\transcript'$.  We note that the first term is at most $\frac{1}{16k \cdot p(n)}$ by Lemma \ref{lem:zk-dist-catch-good}; the third term is at most $\frac{1}{16k \cdot p(n)} + e^{-n}$ by Lemma \ref{lem:zk-p-tilde-good}.  For the second term, we note that since $m_i$ is not ``very good'' but $m_i'$ is ``okay'', the probability that $\Ver$ rejects $\transcript'$ can be at most $\frac{4}{p_\Est(n)}$ higher than the probability that $\Ver$ rejects $\transcript$.  Thus, the second term can be bounded as
        \begin{align*}
        &\P_{\Hyb_i^2}[\Ver(x,\transcript') = 0 \land \overline{R_i} \land m_i \text{ is not ``very good''} \land m_i' \text{ is ``okay''}]\\
            \leq &\P_{\Hyb_i^2}[\Ver(x,\transcript) = 0 \land \overline{R_i} \land m_i \text{ is not ``very good''} \land m_i' \text{ is ``okay''}] + \frac{4}{p_\Est(n)} \\
            \leq &\P_{\Hyb_i^2}[\Ver(x,\transcript) = 0 \land \overline{R_i}] + \frac{4}{p_\Est(n)} = \P_{\Hyb_i^1}[\Ver(x,\transcript) = 0 \land \overline{R_i}] + \frac{4}{p_\Est(n)}
        \end{align*}
        where the equality comes from the fact that $\transcript$ is sampled from the same distribution in the two hybrids.  Plugging this all in, we get that
        \begin{align*}
            \P_{\Hyb_i^2}[\Ver(x,\transcript') = 0 \land \overline{R_i}] \leq \P_{\Hyb_i^1}[\Ver(x,\transcript) = 0 \land \overline{R_i}] + \frac{4}{p_\Est(n)} + \frac{1}{8k \cdot p(n)} + e^{-n}
        \end{align*}
        Noting that $\frac{4}{p_\Est(n)} + \frac{1}{8k \cdot p(n)} + e^{-n} \leq \frac{1}{4k \cdot p(n)}$ for sufficiently large $n$, we get the desired claim.
    \end{proof}
    
\end{proof}

\subsection{Putting It All Together}

\begin{proof}[Proof of Theorem \ref{thm:zk-to-owf-wc}]
    Suppose for the sake of contradiction that ai-OWFs don't exist, and hence $\paramSet$ is infinite.  We then have that Algorithm \ref{alg:zk-alg} satisfies the requirements of Definition \ref{def:io-bpp} with respect to $a(n) = 1 - \ec(n) - \ezk(n) - \frac{1}{2p(n)}$ and $b(n) = \es(n)$.  Indeed, Lemma \ref{lem:zk-alg-correctness} immediately tells us that we satisfy the first two points for infinitely many $n$.  For the third point, we have that for sufficiently large $n$,
    \begin{align*}
        a(n) - b(n) = 1 - \ec(n) - \ezk(n) - \es(n) - \frac{1}{2p(n)} \geq \frac{1}{2p(n)}
    \end{align*}
    where the inequality comes from the definition of $p(n)$.  Thus, we have that $\Lang \in \ioBPPpoly$, so by Remark \ref{rmk:bpppoly-vs-ppoly}, we have that $\Lang \in \ioPpoly$.  This provides our contradiction.
\end{proof}

\begin{remark} \label{rmk:zk-owf-reduction}
    Theorem \ref{thm:zk-to-owf-wc} can also be phrased in the following way: if $\Lang$ has a constant-round, public-coin $(\ec, \es, \ezk)$-ZK argument with $\ec + \es + \ezk <_n 1$, then $\Lang$ has a decision-to-inversion reduction (Definition \ref{def:dti-reduction}).  The reasoning for this is identical to that in Remark \ref{rmk:nizk-owf-reduction}.
\end{remark}

Given the above remark, the techniques from \cite{HirNan,LMP} immediately allow us to improve Theorem \ref{thm:zk-to-owf-wc} from an \emph{auxiliary-input} one-way function to a standard one-way function.

\begin{corollary} \label{cor:zk-to-owf}
    Suppose that every $\Lang \in \NP$ has a constant-round, public-coin $(\ec, \es, \ezk)$-ZK argument with $\ec + \es + \ezk <_n 1$ and $\NP \not \subseteq \ioPpoly$.  Then OWFs exist.
\end{corollary}

\begin{proof}
    Since $\NP \not \subseteq \ioPpoly$, we know there is a language $\Lang \in \NP$ (that thus has an appropriate argument) where $\Lang \not \in \ioPpoly$.  Theorem \ref{thm:zk-to-owf-wc} thus tells us that there exists an auxiliary-input one-way function.  We can then plug this into Lemma \ref{lem:lmp-aiowf-to-avghard} to get a new language $\Lang' \in \NP$ and a polynomial $m(\cdot)$ such that $(\Lang', \{ U_{m(n)} \}_{n \in \mathbb{N}}) \not \in \aBPPpoly$.  But since $\Lang' \in \NP$ (and thus has an appropriate argument), Remark \ref{rmk:zk-owf-reduction} tells us that there is a decision-to-inversion reduction for $\Lang'$.  Hence, Lemma \ref{lem:lmp-reduction-to-owf} tells us that one-way functions exist.
\end{proof}

\section{Applications to Amplification}

Similar to \cite{CHK}, we can plug the one-way function we get from a high-error zero-knowledge protocol into prior work in order to get (almost) unconditional amplification of these high-error protocols.  Below, we separately consider the cases of non-interactive zero-knowledge and interactive zero-knowledge.

\subsection{Non-Interactive Zero-Knowledge}

In the non-interactive case, a recent line of works \cite{GJS19,BKP+24,BG,AK25} has explored what assumptions are needed in order to amplify (adaptively-sound) NIZKs with high errors to have negligible errors, with \cite{BG,AK25} only requiring one-way functions in certain cases.  If we instantiate these results with the one-way function we get from Corollary \ref{cor:nizk-to-owf}, we get that we can amplify such NIZKs (almost) unconditionally.  Formally, we recall the following theorem from \cite{AK25}, rephrased slightly to fit our terminology.

\begin{theorem}[\cite{AK25}, Theorem 8.27] \label{thm:ak-stat-amp}
    Suppose OWFs exist.  Then if there exists an $(\ec, \es, \ezk)$-NIZK {\em proof} with adaptive and statistical soundness for a language $\Lang$ with (1) $\ec$ negligible and (2) $\es + \ezk <_n 1$, then there exists an $(\ec', \es', \ezk')$-NIZK with adaptive and statistical soundness for $\Lang$ with $\ec'$, $\es'$, and $\ezk'$ all negligible.
\end{theorem}

Combining this with Corollary \ref{cor:nizk-to-owf} immediately gives us the following result.

\begin{corollary} \label{cor:nizk-amp-stat}
    Suppose that $\NP \not \subseteq \ioPpoly$ or $\NP \subseteq \BPP$.  Suppose further that all of $\NP$ has $(\ec, \es, \ezk)$-NIZKs with adaptive and statistical soundness where (1) $\ec$ is negligible and (2) $\ezk + \es <_n 1$.  Then all of $\NP$ has $(\ec', \es', \ezk')$-NIZKs with adaptive and statistical soundness where $\ec'$, $\es'$, and $\ezk'$ are all negligible.
\end{corollary}

\begin{proof}
    First, suppose $\NP \subseteq \BPP$.  In this case, every $\Lang$ in $\NP$ has a trivial NIZK ($\Pro$ sends nothing and $\Ver$ checks the instance on their own), meaning the claim holds trivially.

    Otherwise, suppose $\NP \not \subseteq \ioPpoly$.  By Corollary \ref{cor:nizk-to-owf}, we get that one-way functions exist; we can then use Theorem \ref{thm:ak-stat-amp} to amplify the NIZK guaranteed for any $\Lang \in \NP$ to have negligible errors.
\end{proof}

\subsection{Interactive Zero-Knowledge}

In the interactive case, a classic result of \cite{BJY97} shows that one-way functions alone suffice to construct a 4-round interactive zero-knowledge protocol with negligible errors.  
Similar to the non-interactive case, we can instantiate this result using the one-way function we get from Corollary \ref{cor:zk-to-owf} in order to (almost) unconditionally ``amplify'' (actually, just build from scratch)

interactive zero-knowledge.  Formally, we recall the following theorem from \cite{BJY97}.

\begin{theorem}[\cite{BJY97}, Theorem 1.1] \label{thm:bjy-owf-to-zk}
    Suppose OWFs exist.  Then for any $\Lang \in \NP$, there exists a four-round $(\ec, \es, \ezk)$-ZK argument for $\Lang$ with $\ec$, $\es$, and $\ezk$ all negligible.
\end{theorem}

Combining this with Corollary \ref{cor:zk-to-owf} immediately gives us the following result.

\begin{corollary}\label{cor:zk-amp-stat}
    Suppose that $\NP \not \subseteq \ioPpoly$ or $\NP \subseteq \BPP$.  Suppose further that all of $\NP$ has a constant-round, public-coin $(\ec, \es, \ezk)$-ZK argument with $\ec + \ezk + \es <_n 1$.  Then all of $\NP$ has a 4-round $(\ec', \es', \ezk')$-ZK argument with $\ec'$, $\es'$, and $\ezk'$ all negligible.
\end{corollary}

\begin{proof}
    First, suppose $\NP \subseteq \BPP$.  In this case, every $\Lang$ in $\NP$ has a trivial zero-knowledge protocol ($\Pro$ sends nothing and $\Ver$ checks the instance on their own), meaning the claim holds trivially.

    Otherwise, suppose $\NP \not \subseteq \ioPpoly$.  By Corollary \ref{cor:zk-to-owf}, we get that one-way functions exist; we can then use Theorem \ref{thm:bjy-owf-to-zk} to construct the desired zero-knowledge argument for $\NP$.
\end{proof}

\ifnum\anonymous=1
\else
\noindent{\bf Acknowledgments.}
JH, DK, and KT were supported in part by NSF CAREER CNS-2238718, NSF CNS-2247727, a Google
Research Scholar award and an
award from Visa Research. 
Part of this work was done while the authors were visiting the Simons Institute for the Theory of Computing.
\bibliographystyle{alpha}
\addcontentsline{toc}{section}{References}
\bibliography{refs}

\appendix
\section*{Supplementary material}

\smallskip

\noindent \textbf{Disclaimer}

\medskip

\noindent Case studies, comparisons, statistics, research and recommendations are provided ``AS IS'' and intended for informational purposes only and should not be relied upon for operational, marketing, legal, technical, tax, financial or other advice.  Visa Inc. neither makes any warranty or representation as to the completeness or accuracy of the information within this document, nor assumes any liability or responsibility that may result from reliance on such information.  The Information contained herein is not intended as investment or legal advice, and readers are encouraged to seek the advice of a competent professional where such advice is required.
 
These materials and best practice recommendations are provided for informational purposes only and should not be relied upon for marketing, legal, regulatory or other advice. Recommended marketing materials should be independently evaluated in light of your specific business needs and any applicable laws and regulations. Visa is not responsible for your use of the marketing materials, best practice recommendations, or other information, including errors of any kind, contained in this document.
\section{Private to Public Coin in the Absence of ai-OWFs} \label{sec:private-to-public} 
% \begin{theorem}
% If there exists an interactive puzzle where the first $t$ verifier messages are uniformly random strings, then either io-uniform-ai-OWFs exist, or there exists an interactive puzzle where the first $t+1$ verifier messages are uniformly random strings. Additionally, the new verifier runs in time polynomial in the original verifier's runtime.
% \end{theorem}
\begin{theorem}
\label{thm:private-to-public}
If there exists an argument system that satisfies the following properties:
\begin{itemize}
    \item It satisfies correctness up to $\epsilon_c$ error.
    \item It satisfies soundness up to $\epsilon_s$ error.
    \item There exists a simulator such that the output of the simulator is at most $\epsilon_{zk}$ distinguishable from the transcript of an honest execution of the protocol. Note that this is weaker than honest-verifier zero-knowledge since the simulator does not simulate the entire view of the verifier, only the transcript.
    \item For some constant $t\in\mathbb{N}$ the first $t$ verifier messages are uniformly random strings. 
\end{itemize}
Then either io-uniform-ai-OWFs exist, or for every polynomial $q$, there exists an argument system that satisfies the following properties:
\begin{itemize}
    \item It satisfies correctness up to $\epsilon_c + 1/q(\lambda)$ error.
    \item It satisfies soundness up to $\epsilon_s + 1/q(\lambda)$ error.
    \item There exists a simulator such that the output of the simulator is at most $\epsilon_{zk} + 1/q(\lambda)$ distinguishable from the transcript of an honest execution of the protocol. 
    \item The first $t+1$ verifier messages are uniformly random strings. 
\end{itemize} Additionally, the new prover, verifier, and simulator run in time polynomial in the original prover, verifier, and simulator runtimes.
\end{theorem}

\begin{proof}
Without loss of generality, we model the verifier $V$ as sampling its first $t$ messages $r_1, \ldots, r_t$ uniformly at random. Then, before sending the $(t+1)^{\text{th}}$ message, $V$ samples internal randomness $r$. 
For $i \geq t+1$, the $i^{\text{th}}$ verifier message is computed by a deterministic polynomial-time machine $V_i$ that takes as input the transcript $\tau$ and the randomness $r$.   
Additionally suppose there exist polynomials $p_1$ and $p_2$ such that the size of the transcript just before the $(t+1)^{\text{th}}$ verifier message is $p_2(\lambda)$ and $|r| = p_1(\lambda)$.

Define $f(a, x) = V_{t+1}(a, x)$. If $f$ is a $1/2q(\lambda)$-distributional (infinitely-often, uniform-secure) auxiliary-input one-way function we are done, since by \cite{ImpThesis} it may be amplified to construct an (infinitely-often, uniform-secure) auxiliary-input one-way function.  
Suppose not, i.e., there exists a PPT machine $A$ such that for all large enough $\lambda$, for any $a \in \{0,1\}^{p_1(\lambda)}$, and for $x \leftarrow \{0,1\}^{p_2(\lambda)}$,
\[
\{ V_{t+1}(a,x), A(a, V_{t+1}(a,x)) \} \approx_{\frac{1}{2q(\lambda)}} \{ V_{t+1}(a,x), x \}.
\]
where the approximation symbol represents statistical indistinguishability.
In particular,
\begin{align}
\label{eq:ai-owf-dist}
\{ V_{t+1}(a,x), x, A(a, V_{t+1}(a,x)) \} \approx_{\frac{1}{q(\lambda)}} \{ V_{t+1}(a,x), A(a, V_{t+1}(a,x)), x \}. 
\end{align}
We proceed by a sequence of hybrids, where the final hybrid will represent a proof system where the first $t+1$ verifier messages are uniformly random.\\

\noindent \textbf{Hybrid 0:} Identical to the honest protocol.\\

\noindent \textbf{Hybrid 1:} 
The verifier behaves identically to the honest verifier except it appends $A(\tau, V_{t+1}(\tau, r))$ to its $(t+1)^{\text{th}}$ message, where $\tau$ is the transcript of the first $t$ messages.  
This appended value is ignored by all subsequent verifier and honest prover steps.\\

\noindent The protocol preserves correctness since the appended value is ignored by both parties. It preserves soundness since the appended value can be sampled by the malicious prover themselves and is ignored by the verifier. 
The simulator for (honest-verifier) zero-knowledge can run the simulator from Hybrid 0 and then sample $A(V_{t+1}(\tau, r))$ and append to the $(t+1)^{\text{th}}$ verifier message in the simulated transcript. Since the honest transcript distribution in Hybrid 1 can be sampled by sampling an honest transcript in Hybrid 0 and running the same sample-and-append step, switching between Hybrid 0 and Hybrid 1 preserves simulator security.\\

\noindent \textbf{Hybrid 2:} 
Identical to Hybrid 1 except the $(t+1)^{\text{th}}$ message is computed as $V_{t+1}(\tau, r) \| r$, and all subsequent verifier steps use $A(\tau, V_{t+1}(\tau, r))$ in place of $r$. In other words, after computing $V_{t+1}(\tau, r)$ the roles of $r$ and $A(\tau, V_{t+1}(\tau, r))$ are switched.\\

\noindent By \eqref{eq:ai-owf-dist}, swapping $A(a, V_{t+1}(\tau, r))$ and $r$ in the transcript after $V_{t+1}(\tau, r)$ can be detected with advantage at most $\frac{1}{q(\lambda)}$, even by unbounded adversaries. 
Therefore, completeness, soundness, and simulation security are all preserved upto $1/q(\lambda)$ error. For simulation security we crucially use the fact that the simulator only simulates the transcript instead of the full view of $V$, since \eqref{eq:ai-owf-dist} does not hold if given access to the randomness of $A$, which is part of the view of $V$.\\

\noindent \textbf{Hybrid 3:} 
Identical to Hybrid 2 except the $(t+1)^{\text{th}}$ message consists only of the randomness $r$.  
$A(a, V_{t+1}(\tau, r))$ is computed before the $(t+2)^{\text{th}}$ verifier message instead of before the $(t+1)^{\text{th}}$ message.\\

\noindent Since $V_{t+1}(\tau, r)$ can be computed from the transcript, this change does not affect completeness, soundness, or simulation security.
The verifier now samples the first $t+1$ messages uniformly at random, so we are done.
\end{proof}
\begin{corollary}
\label{cor:private-to-public}
    The existence of a constant-round $(\epsilon_c, \epsilon_s, \epsilon_{zk})$ (honest-verifier) zero-knowledge argument implies either io-uniform-ai-OWFs exist, or for every polynomial $q$, there exists a \textit{public-coin} constant-round $(\epsilon_c + 1/q(\lambda), \epsilon_s + 1/q(\lambda), \epsilon_{zk} + 1/q(\lambda))$ (honest-verifier) zero-knowledge argument.
\end{corollary}
\begin{proof}
    Follows from repeated application of Theorem \ref{thm:private-to-public} along with the following observations:
    \begin{itemize}
        \item While there is polynomial blowup in runtime each time the theorem is applied, since there are constant rounds the final runtimes are polynomial.
        \item While errors grow by an inverse polynomial additive term, we can set the to be small enough to obtain our desired final bounds.
        \item While the theorem only guarantees a weakening of honest zero-knowledge, for public-coin protocols, both notions are the same. Therefore we recover honest-verifier zero-knowledge for the final protocol. 
    \end{itemize}
\end{proof}

\section{Proof of Lemma \ref{lem:lmp-reduction-to-owf}} \label{sec:lmp-proof}

Here we provide an explicit proof of Lemma \ref{lem:lmp-reduction-to-owf}, noting that it is closely mirrors the proof of Lemma 3.1 in \cite{LMP}.

\begin{proof}[Proof of Lemma \ref{lem:lmp-reduction-to-owf}]
    Let $(R^{(\cdot)}, \{f_x\}_{x \in \zo^*}, p(\cdot))$ be the decision-to-inversion reduction for $\Lang$ guaranteed by the statement of the lemma.  We define the function $g(r_1\|r_2) = \Dist(r_1) \| f_{\Dist(r_1)}(r_2)$, and claim that $g$ is a weak one-way function.

    Indeed, suppose for the sake of contradiction that $g$ is not a weak one-way function.  This means that there exists an algorithm $\Alg$ that inverts $g$ with probability at least $0.75$ on infinitely many input sizes.  Using $\Alg$, we define an algorithm $\mathcal{B}$ to decide $(\Lang, \Dist)$ as follows.  On input $x$, $\mathcal{B}$ first samples a random $r_2$ and runs $\Alg(x, f_x(r_2))$.  If $\Alg$ fails to find a preimage of $f_x(r_2)$, $\mathcal{B}$ immediately halts and outputs $\fail$.  Otherwise, $\mathcal{B}$ outputs $R^\Alg(x)$.

    We now claim that $\mathcal{B}$ puts $(\Lang, \Dist) \in \aBPPpoly$, which is a contradiction.  Fix any $n$ where $\Alg$ succeeds in inverting $g$ and that is sufficiently large for our guarantees on $R$ to hold.  To show that the first point of Definition \ref{def:ioAvgBPP} holds, we fix an $x \in \Supp(\Dist_n)$ and consider two possible cases.  If $\Alg(x, \cdot)$ inverts $f_x$ with probability at most $\frac{1}{p(n)}$, we have that the first step of $\mathcal{B}(x)$ will cause it to output $\fail$ with probability at least $1 - \frac{1}{p(n)}$.  If instead $\Alg(x, \cdot)$ inverts $f_x$ with probability at least $\frac{1}{p(n)}$, our guarantees on $R$ tell us that (as long as the first step doesn't output $\fail$), $\mathcal{B}(x)$ will output $\Lang(x)$ with probability at least $1 - \frac{1}{p(n)}$.  Thus in either case, we have that $\mathcal{B}(x) \in \{ \fail, \Lang(x) \}$ with probability at least $1 - \frac{1}{p(n)}$, which for sufficiently large $n$ is at least $0.9$.

    For the second point of Definition \ref{def:ioAvgBPP}, we note that $\mathcal{B}(x)$ can only output $\fail$ if $\Alg(x, \cdot)$ fails to invert $f_x$.  Thus the probability of this happening over a random $x \samp \Dist_n$ is exactly the probability that $\Alg$ fails to invert a random output of $g$, which we know to be at most $0.25$.

    Putting this all together, we have reached a contradiction, and so must have that $g$ is in fact a weak one-way function.  By Theorem \ref{thm:wowf-vs-owf}, this means that one-way functions exist.
\end{proof}

\end{document}
\typeout{get arXiv to do 8 passes: Label(s) may have changed. Rerun}